\documentclass[11pt]{article}
\usepackage[utf8]{inputenc}
\usepackage{amssymb,amsthm,amsfonts}
\usepackage{amsmath}
\usepackage{graphicx}
\usepackage[dvipsnames]{xcolor}
\usepackage{fp}
\usepackage{algorithm}
\usepackage[noend]{algpseudocode}
\usepackage{fullpage}
\usepackage{epsfig}
\usepackage{todonotes}
\usepackage{xspace}
\usepackage{framed, mdframed}
\usepackage{caption}
\usepackage{environ}
\usepackage{thmtools, thm-restate}

\algrenewcommand\textproc{}
\makeatletter
\NewEnviron{Lalign}{\tagsleft@true\begin{align}\BODY\end{align}}
\makeatother

\newtheorem{theorem}{Theorem}
\newtheorem{definition}[theorem]{Definition}
\newtheorem{lemma}[theorem]{Lemma}

\newtheorem{corollary}[theorem]{Corollary}
\newtheorem{remark}[theorem]{Remark}

\newcommand{\req}[1]{(\ref{#1})}

\newcommand{\StrongInput}{strong}
\newcommand{\WeakInput}{weak}
\newcommand{\StrongPart}{strong}
\newcommand{\WeakPart}{weak}
\newcommand{\GeneralProblem}{Hierarchical~Correlation~Clustering}
\newcommand{\NewProblem}{Hierarchical~Cluster~Agreement}

\newcommand{\HCC}{HierCorrClust}
\newcommand{\HCA}{HierClustAgree}
\newcommand{\CC}{CorrClust}
\newcommand{\dist}{\text{dist}}
\newcommand{\calD}{\mathcal{D}}
\newcommand{\calE}{\mathcal{E}}
\newcommand{\calF}{\mathcal{F}}

\newcommand{\calO}{O}
\newcommand{\FOPT}{x^{(*)}}

\def \closeToZero {0.1}
\def \radiusInCluster {\closeToZero}
\def \dLC {0.1}
\def \dPlus {0.3}
\def \dPlusOne {1.3}
\def \outsideL {\frac{1}{6}}
\def \diameterInCluster {\FPeval\calc{clip(trunc(2*\radiusInCluster:6))}\calc}
\def \twoDiameter {\FPeval\calc{clip(trunc(4*\radiusInCluster:6))}\calc}
\def \threeDiameter {\FPeval\calc{clip(trunc(6*\radiusInCluster:6))}\calc}
\def \badOutEdges {0.05}
\def \deltaInLp {\FPeval\calc{clip(trunc(1-\closeToZero:6))}\calc}

\def \TNI {top-non-intersecting descendants}

\newcommand{\lpd}[1]{LP-distance at level $#1$}
\newcommand{\B}[2]{B_{<#1}^{(#2)}}
\newcommand{\DelPlus}[1]{\Delta^+(#1)}

\title{Fitting Distances by Tree Metrics Minimizing the Total Error within a Constant Factor}

\author{
  Vincent Cohen-Addad\thanks{Google Research, Switzerland. Email: \texttt{\{cohenaddad,nikosp\}@google.com}}
  \and
  Debarati Das\thanks{Department of Computer Science, University of Copenhagen. Research of this author is supported by VILLUM Investigator Grant 16582, Basic Algorithms Research Copenhagen (BARC), Denmark. Email: \texttt{\{debaratix710,mikkel2thorup\}@gmail.com, kipouridis@di.ku.dk}\includegraphics[height=0.02\textwidth, width=0.035\textwidth]{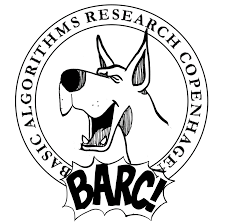}}
  \and
  Evangelos Kipouridis\footnotemark[2] \thanks{Evangelos Kipouridis has received funding from the European Union’s Horizon 2020 research and innovation programme under the Marie Skłodowska-Curie grant agreement No 801199.\includegraphics[height=0.02\textwidth, width=0.035\textwidth]{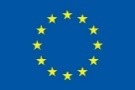}}
  \and
  Nikos Parotsidis\footnotemark[1]
  \and
  Mikkel Thorup\footnotemark[2]
}

\date{}

\begin{document}

\maketitle
\thispagestyle{empty}
\begin{abstract}
\begin{sloppypar}
We consider the numerical taxonomy problem of fitting a positive distance function ${\calD:{S\choose 2}\rightarrow \mathbb R_{>0}}$ by a tree metric. We want a tree $T$ with positive edge weights and including $S$ among the vertices so that their distances in $T$ match
those in $\calD$. A nice application is in evolutionary biology
where the tree $T$ aims to approximate the branching process leading
to the observed distances in $\calD$ [Cavalli-Sforza and Edwards 1967].
We consider the total error, that is the sum of distance errors over
all pairs of points.  We present a deterministic polynomial time algorithm
minimizing the total error within a constant factor. We can do this
both for general trees, and for the special case of ultrametrics with
a root having the same distance to all vertices in $S$.
\end{sloppypar}

The problems are APX-hard, so a constant factor is the best we can hope for in polynomial time.
The best previous approximation factor was
$\calO((\log n)(\log \log n))$ by Ailon and Charikar [2005] who wrote ``Determining whether an
$\calO(1)$ approximation can be obtained is a fascinating question".
\end{abstract}

\clearpage
\pagenumbering{arabic}
\section{Introduction}
\label{sec:intro}

Taxonomy or hierarchical classification of species goes back at
least to discussions between Aristotle and his teacher Plato%\footnote{``Classification", https://iep.utm.edu/classifi/, The Internet Encyclopedia of Philosophy}
\footnote{https://iep.utm.edu/classifi/, Internet Encyclopedia of Philosophy}
($\sim$350BC)
while modern taxonomy is often attributed to
Linnaeus%\footnote{``The Linnaean system", https://britannica.com/science/taxonomy/The-Linnaean-system} ($\sim$1750). 
\footnote{https://britannica.com/science/taxonomy/The-Linnaean-system} ($\sim$1750).
The discussions of evolution in the nineteenth century, clarified the
notion of evolutionary trees, or phylogenies, and the notion
that species were close due
to a common past ancestor. Such evolutionary trees
are seen in the works of Hitchcock%\footnote{``Edward Hitchcock", https://en.wikipedia.org/wiki/Edward\_Hitchcock}
\footnote{https://en.wikipedia.org/wiki/Edward\_Hitchcock}
(1840)
and Darwin%\footnote{``On the Origin of Species", https://en.wikipedia.org/wiki/On\_the\_Origin\_of\_Species} (1859).
\footnote{https://en.wikipedia.org/wiki/On\_the\_Origin\_of\_Species} (1859).
Viewing the descendants of each node as a class, the
evolutionary tree induces a hierarchical classification.

In the 1960s came the interest in computing evolutionary trees based
on present data, the so-called numerical taxonomy problem \cite{4-cavalli67, sneath1962numerical,7-Sneath-Numerical-1963}.  Our focus is on the following simple model by Cavalli-Sforza
and Edwards from 1967 \cite{4-cavalli67}. In an evolutionary tree,
let the edge between the child and
its parent be weighted by the evolutionary distance between them.
Then the evolutionary distance between any two species is the sum of
weights on the unique simple path between them.  We note that the selection of the root
plays no role for the distances. What we are saying is that
any tree with edge weights, induces distances between its nodes; a
so-called tree metric assuming that all weights are positive.

We now have the converse reconstruction problem of numerical taxonomy \cite{4-cavalli67, sneath1962numerical, 7-Sneath-Numerical-1963}:
given a set $S$ of species with measured distances between them, find a
tree metric matching those observed distances on $S$. Thus we are looking for an
edge-weighted tree $T$ which includes $S$ among its nodes with the right
distances between them. Importantly, $T$ may have nodes not in $S$
representing ancestors explaining proximity between different species.
The better the tree metric $T$ matches the measured distances on $S$, the better the tree $T$ explains these measured distances.

\paragraph{Other applications}
This very basic reconstruction problem also arises in various other contexts.
First, concerning the evolutionary model, it may be considered too simplistic to
just add up distances along the paths in the tree. Some changes from parent
to child could be reverted for a grandchild. Biologists \cite{CAVENDER1978271, Far72}
have suggested stochastic models for probabilistic changes that also have a chance
of being reverted further down. However, Farach and Kannan \cite{DBLP:journals/jacm/FarachK99} have shown that applying
logarithms appropriately, we can convert estimated distances into
some other distances for which we find a matching tree metric that we can then
convert back into a maximum likelihood stochastic tree.  
The basic point is that finding tree metrics can be used as powerful tool to invert evolution even in 
cases where tree metric model does not apply directly.

Obviously, the numerical taxonomy problem is equally relevant to other historical sciences
with an evolutionary branching process leading to evolutionary distances, e.g., historical
linguistics.

More generally, if we can
approximate distances with a tree metric, then the tree of this metric
provides a very compact and convenient representation that is much easier to
navigate than a general distance function. Picking
any root, the tree induces a distance based hierarchical classification, 
and referring to the discussions between Plato and Aristotle's, humans have been 
interested in such hierarchical classifications since ancient time. 

It is not just humans but also computers that understand trees and tree metrics much better than general metrics. Many questions that are NP-hard to answer in general
can be answered very efficiently based on trees (see, e.g., Chapter 10.2 ``Solving NP-Hard Problems on Trees" in the text book \cite{AlgorithmDesign}).

Computing ``good"
tree representations is nowadays also a major tool to learn from data.
In this context, we are sometimes interested in a special kind
of tree metrics, called \emph{ultrametrics}, defined by rooted trees whose sets of leaves is $S$ and where the
leaf-to-root distance is the same for all points in $S$.
Equivalently, an ultrametric is a metric so that for any three points $i,j,k$, the distance from $i$ to $j$ is no bigger than the maximum of the distance from $i$ to $j$ and the distance from $j$ to $k$\ %\footnote{``Ultrametric space", https://en.wikipedia.org/wiki/Ultrametric\_space.}.
\footnote{https://en.wikipedia.org/wiki/Ultrametric\_space.}.

An ultrametric can be seen as modeling
evolution that is linear over time. This may be not the case in biology where the speed of evolution
depends on the local evolutionary pressure for example. However, ultrametrics are key objects in machine
learning and data analysis, see e.g.: \cite{carlsson2010characterization}, and there are various algorithms for embedding arbitrary 
metrics into ultrametrics such as the popular ``linkage" algorithms (single, complete or average linkage), see also~\cite{Cohen-AddadKMM19, MoseleyW17}.

\subsection{Tree fitting (Numerical Taxonomy Problem)}
Typically our measured distances do not have an exact fit with any tree metric. We then have
the following generic optimization problem for any $L_p$-norm:\smallskip

\noindent{\bf Problem:} $L_p$-fitting tree (ultra) metrics\smallskip\\
\noindent{\bf Input:} A set $S$ with a distance
function $\calD:{S\choose 2}\rightarrow\mathbb R_{>0}$.  
 \footnote{${S\choose k}$ denotes all subsets of $S$ of size $k$.} \smallskip\\
\noindent{\bf Desired Output:} A tree metric (or ultrametric) $T$ that spans $S$ and fits $\calD$ in the sense of minimizing the $L_p$-norm
\begin{equation}\label{eq:objective}
\|T-\calD\|_p=\left(\sum_{\{i,j\}\in {S \choose 2}} |\dist_T(i,j)-\calD(i,j)|^p\right)^{1/p}.
\end{equation}
 Cavalli-Sforza and Edwards~\cite{4-cavalli67} introduced this numerical taxonomy problem for both tree and ultrametrics in the $L_2$-norm in 1967. Farris suggested using $L_1$-norm in 1972 \cite[p. 662]{Far72}.

\subsection{Our result}
In this paper we focus on the $L_1$-norm, that is, the total sum of errors.
The problem is APX-hard, both for tree metrics and ultrametrics (see Section~\ref{sec:apx} and also
\cite{DBLP:journals/siamcomp/AilonC11}), so a constant approximation factor is the best we can hope for in polynomial time.
The best previous approximation factor for both tree metrics and ultrametrics was
$\calO((\log n)(\log \log n))$ by Ailon and Charikar \cite{DBLP:journals/siamcomp/AilonC11}.

In this paper we present a deterministic polynomial time constant factor
approximation both for tree metrics and for ultrametrics, that is, in both cases, we
can find a tree $T$ minimizing the $L_1$-norm within a constant factor of the best possible.

Thus, we will prove the following theorem.

\begin{theorem} \label{thm:main-result}
The $L_1$-fitting tree metrics problem can be solved in deterministic polynomial time within a constant approximation factor. The same holds for the $L_1$-fitting ultrametrics problem.
\end{theorem}

\subsection{History of $L_p$ tree fitting}
Since Cavalli-Sforza and Edwards introduced the tree fitting problem, the problem has collected an extensive literature. In 1977 \cite{10-WATERMAN1977199}, it was
shown that if there is a tree metric coinciding exactly with $\calD$, it
is unique and it can be found in time linear in the input size, i.e., $\calO(|S|^2)$ time.
The same then also holds trivially for ultrametrics. Unfortunately
there is typically no tree metric coinciding
exactly with $\calD$, and in 1987 \cite{5-DAY1987461} it was shown that for $L_1$ and $L_2$ the
numerical taxonomy problem is NP-hard, both in the tree metric and 
the ultrametric cases. The problems are in fact APX-hard (see Section~\ref{sec:apx}), which rules
out the possibility of a polynomial-time approximation scheme. Thus, a constant
factor, like ours for $L_1$, is the best one can hope for from a complexity perspective for these
problems.

For the $L_\infty$ numerical taxonomy problem, there was much more progress. In 1993 \cite{6-DBLP:journals/algorithmica/FarachKW95} it was
shown that for the ultrametric case an optimal solution can be found in
time proportional to the number of input distance
% \mtcom{we say distance, not dissimilarity.. also, above we call it $\calO(|S|^2)$.. let's stay consistent unless it really is different.. E: it works for non-complete matrices in time proportional to the number of known entries (plus $|S|log|S|$)}
pairs (i.e.: the number of entries in $S$). More recently, it was shown that when the points are 
embedded into $\mathbb{R}^d$ and the distances are given by the pairwise Euclidean distances, the problem can be approximated in subquadratic 
time~\cite{Cohen-AddadSL20,Cohen-AddadJL21}. For the general trees case (still in the $L_\infty$-norm objective), 
\cite{DBLP:journals/siamcomp/AgarwalaBFPT99} gave
an $\calO(|S|^2)$ algorithm that produces a constant factor approximation and
proved that the problem is APX-hard (unlike the ultrametric case).

The technical result from \cite{DBLP:journals/siamcomp/AgarwalaBFPT99} was a general reduction from
general tree metrics to ultrametrics. It modifies the input distance matrix
and asks for fitting this new input with an ultrametric that can later be converted to a tree metric
for the original distance matrix.  The result states that for any $p$, if we can  minimize the restricted ultrametric $L_p$ error within a factor $\alpha$ in polynomial-time,
then there is a polynomial-time algorithm that minimizes the tree metric $L_p$ error within a factor $3\alpha$. The reduction from
\cite{DBLP:journals/siamcomp/AgarwalaBFPT99} imposes a certain restriction on the ultrametric, but
the restriction is not problematic and in Section \ref{sec:tree-3ultra}, we will even argue that the restriction can be completely eliminated with a slightly modified reduction.
With $n$ species, the reduction from tree metrics to ultrametrics can be performed in time $\calO(n^2)$. Applying this to the optimal ultrametric algorithm from \cite{6-DBLP:journals/algorithmica/FarachKW95}
for the $L_\infty$-norm objective yielded a factor $3$ for general metrics for
the $L_\infty$-norm objective. The generic impact is that \emph{for any $L_p$, later
  algorithms only had to focus on the ultrametric case to immediately get
  results for the often more important tree metrics case, up to losing a factor
  3 in the approximation guarantee}. Indeed,
  the technical result of this paper is a constant factor approximation for ultrametric.
Thus, when
  it comes to approximation factors, we have
\begin{equation*}
\text{TreeMetric}\leq 3\cdot \text{UltraMetric}
\end{equation*}

For $L_p$ norms with constant $p$, 
the developments have been much slower. Ma et al. \cite{15ac-DBLP:journals/jco/MaWZ99}  considered the problem of finding the best
$L_p$ fit by an ultrametric where distances in the ultrametric
are no smaller than the input distances. For this problem, they
obtained an $\calO(n^{1/p})$ approximation.

Later, Dhamdhere \cite{6ac-DBLP:conf/approx/Dhamdhere04} considered the problem of finding a line metric to minimize additive distortion from the given data
(measured by the $L_1$ norm) and obtained an $\calO(\log n)$ approximation. In fact, his motivation for considering this
problem was to develop techniques that might be useful for finding the closest tree metric with distance measured by the
$L_1$ norm. Harp, Kannan and McGregor \cite{12ac-DBLP:conf/approx/HarbKM05} developed a factor
$\calO(\min\{n,k\log n\}^{1/p})$ approximation for the closest ultrametric under the $L_p$ norm where $k$ is the number
of distinct distances in the input\footnote{The authors erroneously claim that they get the same approximation for the closest tree metric problem. However, the known reduction may create $\omega(k)$ distinct distances.}.

The best bounds known for the ultrametric variant of the problem are due to Ailon and Charikar \cite{DBLP:journals/siamcomp/AilonC11}. 
They first focus on ultrametrics in $L_1$ and show
that if the distance matrix has only $k$ distinct distances, then it is possible to approximate the $L_1$ error within a factor $k+2$. Next they obtain an LP-based $\calO((\log n)(\log\log n))$ approximation for arbitrary distances matrices. Finally
they sketch how it can be generalized to an $\calO(((\log n)(\log\log n))^{1/p})$ approximation of the $L_p$ error for any
$p$. Using the reduction from \cite{DBLP:journals/siamcomp/AgarwalaBFPT99}, they also get an  $\calO(((\log n)(\log\log n))^{1/p})$ approximation for tree metrics under the $L_p$-norm objective. The 
$\calO(((\log n)(\log\log n))^{1/p})$ approximation 
comes from an $\calO((\log n)(\log\log n))$ approximation 
of the $p$'th moment $F_p$:
\begin{equation}\label{eq:moment}
\|T-\calD\|_p^p=\left(\sum_{\{i,j\}\in {S \choose 2}} |\dist_T(i,j)-\calD(i,j)|^p\right).
\end{equation}
Technically, Ailon and Charikar \cite{DBLP:journals/siamcomp/AilonC11} present a
simple LP relaxation for $L_1$ ultrametric fitting---an LP that will also be used in our paper. They get their
$\calO((\log n)(\log\log n))$ approximation using an LP rounding akin
to the classic $\calO(\log n)$ rounding of Leighton and
Rao for multicut 
\cite{DBLP:journals/jacm/LeightonR99}. The challenge
is to generalize the approach to deal with the hierarchical issues associated with ultrametric and show that this can be done paying only an extra factor $\calO(\log\log n)$ in the approximation factor.
Next they show that their LP formulation and rounding is general enough to handle different variants, including other $L_p$ norms as mentioned above, but
also they can handle the \emph{weighted case}, where for each pair of species $i,j$, 
the error contribution to the overall error is multiplied by a value $w_{ij}$. However, this weighted problem captures the multicut problem
(and the weighted minimization correlation clustering problem) \cite{DBLP:journals/siamcomp/AilonC11}. Since the multicut cannot be approximated within a 
constant factor assuming the unique games conjecture \cite{DBLP:journals/cc/ChawlaKKRS06} and the best known approximation bound remains $\calO(\log n)$, it is beyond reach of current
techniques to do much better in these more general settings. 

Ailon and Charikar \cite{DBLP:journals/siamcomp/AilonC11} conclude that
``Determining
whether an $\calO(1)$ approximation can be obtained is a fascinating question. The LP formulation used in our [their] work could eventually lead to such a result''.
For their main LP formulation for the (unweighted) $L_1$ ultrametric fitting,the integrality gap was only known to be somewhere between $2$ and $\calO((\log n) (\log \log n))$. 
To  break the $\log n$-barrier we must come up with a radically different way of rounding this LP and free ourselves from the multicut-inspired approach.

For $L_1$ ultrametric fitting, we 
give the first constant factor approximation, and we show this can be obtained by rounding the LP proposed
by Ailon and Charikar, thus demonstrating a constant integrality gap for their LP. Our solution breaks the $\log n$ barrier using the special combinatorial structure of the
$L_1$ problem.

Stepping a bit back, having different solutions for different norms should not come as a surprise. As an analogue, take the generic problem of placing $k$ facilities in such a way that each of $n$ cities is close to the nearest facility. Minimizing the vector of distances in the $L_1$ norm is called the $k$-median problem. In the $L_2$ norm it is called the $k$-means problem, and in the $L_\infty$-norm is called the $k$-center
    problems. Indeed, while the complexity of the $k$-center problem has been settled in the mid-1980s thanks to Gonzalez' algorithm~\cite{gonzalez1985clustering}, it has remained
    a major open problem for the next 15 years to obtain constant factor approximation for the $k$-median and the $k$-means problems. Similarly, our 
    understanding of the $k$-means  problems ($L_2$-objective) remains poorer than our understanding of the $k$-median problem, and the problem is in fact 
    provably harder (no better than $1+8/e$-approximation algorithm~\cite{guha1999greedy} while $k$-median can be approximated within a factor $2.675$ at the moment~\cite{byrka2014improved}).
    
    For our tree fitting problem, the $L_\infty$ norm has been understood since the 1990s, and our result shows that 
    the $L_1$ norm admits a constant factor approximation algorithm. The current
    status of affairs for tree and ultrametrics is summarized in  Table \ref{tab:Lp}. 
    The status for $L_p$ tree fitting is that we have good constant factor approximation if we want to minimize the total error $L_1$ or the maximal error $L_\infty$. For all other
    $L_p$ norms, we only have the much weaker but general $O(((\log n)(\log\log n))^{1/p})$ approximation from \cite{DBLP:journals/siamcomp/AilonC11}. 
    In particular, we do not know if anything better is possible with $L_2$. The difference is so big that even if we are in a situation where we would normally prefer an $L_2$ approximation, our much better approximation guarantee with $L_1$ might be preferable.

\begin{table}
\begin{center}
\begin{tabular}{|l|c|c|c|}\hline
Norm     & $L_1$ & $L_p$, $p<\infty$ & $L_\infty$\\\hline
Treemetric     & $\Theta(1)$& $\calO(((\log n)( \log \log n))^{1/p})$&$\Theta(1)$\\
Ultrametric     & $\Theta(1)$& $\calO(((\log n)(\log \log n))^{1/p})$ & 1\\\hline
\end{tabular}
\end{center}
\caption{Tree fitting approximation factors.}\label{tab:Lp}
\end{table}

\subsection{Other related work}
\paragraph{Computational Biology.}
Researchers have also studied reconstruction of phylogenies under stochastic models of evolution 
%\mtcom{Is this not related to what Farach and Kannan did, and if so, should we not refer back?}
(see Farach and Kannan~\cite{DBLP:journals/jacm/FarachK99} 
or Mossel
et al. \cite{16ac-DBLP:conf/stoc/MosselR05} and the references therein, see also Henzinger et al.~\cite{DBLP:journals/algorithmica/HenzingerKW99}).

Finally, related to our hierarchical correlation clustering problem is the 
hierarchical clustering problem 
introduced by Dasgupta~\cite{dasgupta15} 
where the goal
is, given a similarity matrix, to build a hierarchical clustering tree where
the more similar two points are, the lower in the tree they are separated (formally, a pair $(u,v)$ induces a cost of similarity($u,v$) times the size of the minimal subtree containing both $u$ and $v$, the goal
is to minimize the sum of the costs of the pairs).
This has received a lot of attention in the last few years (\cite{AlonAV20, CC17, CCN19, CCN18, ChatziafratisYL20, Cohen-AddadKM17, Cohen-AddadKMM19, MW17, RP16}, see also~\cite{AbboudCH19, Balcan08, ChamiGCR20, cochez2015twister}), but differs from our settings since the resulting tree may not induce a metric space.

% \mtcom{We have to have a clear division between works addressing numerical taxonomy: finding a singe evolutionary tree, and other things like Bartal's work which has nothing to do with numerical taxonomy}
\paragraph{Metric Embeddings.}
There is a large body of work of metric embedding problems. 
For example, the metric violation distance problem asks to embed an arbitrary
distance matrix into a metric space while minimizing the $L_0$-objective (i.e.:
minimizing the number of distances that are not preserved in the metric space). 
The problem is considerably harder and is only known to admit an
$\calO(\text{OPT}^{1/3})$-approximation algorithm~\cite{FanGRSB20, FanRB18, GilbertJ17} while
no better than a 2 hardness of approximation is known.
More practical results on this line of work includes~\cite{DBLP:conf/nips/SonthaliaG20} and~\cite{DBLP:conf/allerton/GilbertS18}.
Sidiropoulos et al~\cite{sidiropoulos2017metric} also considered
the problem of embedding into metric, ultrametric, etc. while minimizing the 
total number of outlier points.

There is also a rich literature on metric embedding problems where the measure of interest is the multiplicative distortion,
and the goal of the problem is to approximate the absolute distortion of the metric space  (as opposed to approximating the optimal embedding of the metric space). Several
such problems have been studied in the context of approximating metric spaces via tree metrics (e.g. \cite{DBLP:conf/focs/Bartal96,8ac-DBLP:journals/jcss/FakcharoenpholRT04}).
% \mtcom{Cite Bartal as well, as it is is his approach... the other is just the final improvement.}
The objective of these works is very different since they are focused on the absolute expected multiplicative distortion over all input metrics while we aim at approximating the \emph{optimal} expected additive distortion for each individual input metric. 

While these techniques have been very successful for designing approximation algorithms for various problems in a variety of contexts, they are not aimed at numerical
taxonomy. Their goal is to do something for general metrics. However, for our tree-fitting problem, the idea is that the ground-truth is a tree, e.g., a phylogeny, and that the distances measured, despite noise and imperfection of the model, are 
close to the metric of the true tree. To recover an approximation to the true tree, we
therefore seek a tree that compares well against the best possible fit of a tree metric.

% \mtcom{Our approach is much more relevant when we are hoping that there is a tree metric that is close to the input metric.... (I thin}
% \mtcom{With all the above other related works, if possible, I think it is good to say why ours is more relevant for our problem, e.g., relative to Bartal, our work is more relevant if we think there is a tree metric that is close to the input, e.g., a tree metric plus noise}

\subsection{Techniques}

We will now discuss the main idea of
our algorithm. Our solution will move through
several combinatorial problems that code
different aspects of the $L_1$-fitting of ultrametrics, but which do not generalize nicely to other norms.

Our result follows from a sequence of constant-factor-approximation 
reductions between problems. To achieve our final goal, we introduce
several new problems that have a key role in the sequence of reductions.
Some of the reductions and approximation bounds have already been extensively studied (e.g.: Correlation Clustering). A roadmap of this sequence of results is given in Figure~\ref{fig:roadmap}. %Moreover, a pseudocode roadmap is presented in Algorithm~\ref{algo:roadmap}. 

\begin{figure}
\begin{center}
  \begin{tabular}{|l|l|} \hline
    \multicolumn{1}{|c|}{\bf Problem}                             & \multicolumn{1}{|c|}{\bf Introduced in} \\ \hline
    TreeMetric                          & Beginning of Section \ref{sec:intro} \\ \hline
    UltraMetric                         & Beginning of Section \ref{sec:intro} \\ \hline
    Correlation Clustering (\CC)              & Section \ref{sec:CorrClust} \\ \hline
    Hierarchical Correlation Clustering (\HCC) & Section \ref{sec:HierCorrClust} \\ \hline
    Hierarchical Cluster Agreement (\HCA)      & Section \ref{sec:HierCorrAgree} \\ \hline
    
  \end{tabular}
  %\captionof{table}{List of problems \label{tab:problems}}
  ~\\~\\~\\
 
  \begin{align}
  \text{TreeMetric} &\le (3+o(1)) \cdot \text{UltraMetric} \tag{A}\label{eq:tree-3ultra}\\
    \text{UltraMetric}  &\le \text{\HCC} \tag{B} \label{eq:UltraMetric-HierCorrClust}\\
    \text{\HCC} &< (\text{\CC}+1) (\text{\HCA}+1) \tag{C} \label{eq:reduction_HCC_to_HCA}\\
    \text{\CC} &= \calO(1) \tag{D} \label{eq:corrclust} \\
    \text{\HCA} &= \calO(1) \tag{E} \label{eq:HCA}
  \end{align}
  
{\bf Approximation factors.} Abbreviated problem names used as approximation factors.
  %\captionof{figure}{Sequence of inequalities leading to our result. We use the names of the problems from Table~\ref{tab:problems} to denote the corresponding approximation factor.}
\end{center}

\begin{algorithmic}[1]

\Procedure{TreeMetric}{$S,\calD$} \Comment{See Section~\ref{sec:tree-3ultra}}
    \State Reduction to UltraMetric based on \cite{DBLP:journals/siamcomp/AgarwalaBFPT99}
\EndProcedure
\\
\Procedure{UltraMetric}{$S,\calD$} \Comment{See Section~\ref{ssec:ultrametric_to_hcc}}
    \State Reduction to \HCC{} based on \cite{DBLP:journals/siamcomp/AilonC11, 12ac-DBLP:conf/approx/HarbKM05}
\EndProcedure
\\
\Procedure{\HCC}{$S,E^{(*)},\delta^{(*)}$} \Comment{NEW. See Section~\ref{sec:hcc_to_hca}}
    \State \textbf{for} $t\in [\ell]$ \textbf{do} $Q^{(t)} \gets$ \CC$(E^{(t)})$
    \State \Return \HCA($S, Q^{(*)}, \delta^{(*)})$
\EndProcedure
\\
\Procedure{\CC}{$S,E$} \Comment{See Section~\ref{sec:CorrClust}}
    \State Use Algorithm from \cite{DBLP:journals/ml/BansalBC04}
\EndProcedure
\\
\Procedure{\HCA}{$S,Q^{(*)},\delta^{(*)}$} \Comment{NEW. See Section~\ref{sec:Hierarchical_Cluster_Agreement}}
    \State $x^{(*)}\gets$ Solve(LP-relaxation($S,Q^{(*)},\delta^{(*)}$))
    \State $L^{(*)} \gets$ LP-Cleaning$(S, Q^{(*)}, x^{(*)})$
    \State \Return Derive-Hierarchy$(S, L^{(*)})$
\EndProcedure

\end{algorithmic}
\caption{Roadmap leading to our result for $L_1$-fitting tree metrics.}
\label{fig:roadmap}
\end{figure}

\subsubsection{Correlation Clustering}\label{sec:CorrClust}
Our algorithms will use a subroutine for what is known as the unweighted minimizing-disagreements correlation clustering problem on complete graphs \cite{DBLP:journals/ml/BansalBC04}. We simply refer to this problem as Correlation Clustering throughout the paper.

First, for any family $P$ of disjoint subsets of 
$S$, let
\[\calE(P)=\bigcup_{T\in P}{T\choose 2}\]
Thus $\calE(P)$ represents the edge sets over
an isolated clique over each set $T$ in $P$. Often $P$ will be a partition of $S$, that is, $\bigcup P=S$.

The  \emph{correlation clustering} takes as
input an edge set $E\subseteq {S\choose 2}$ and seeks a partition $P$ minimizing 
\[|E\Delta\calE(P)|\textnormal,\]
where $\Delta$ denotes symmetric difference. It is well-known that correlation clustering is equivalent to ultrametric fitting with two distances (see, e.g., \cite{12ac-DBLP:conf/approx/HarbKM05}).

A randomized polynomial time $2.06+\epsilon$ factor approximation from \cite{DBLP:conf/stoc/ChawlaMSY15} (see also~\cite{DBLP:journals/jacm/AilonCN08}) and a 2.5 deterministic approximation
algorithm~\cite{DBLP:journals/mor/ZuylenW09} 
are known. We shall use
this as a subroutine with approximation factor
\begin{equation*}
    \text{\CC}=\calO(1) \tag*{(\ref{eq:corrclust}) from Figure~\ref{fig:roadmap}}
\end{equation*}
We note that Ailon and Charikar, who presented
the previous best $\calO((\log n)(\log\log n))$ approximation for tree metrics and ultrametrics at
FOCS'05 \cite{DBLP:journals/siamcomp/AilonC11} had presented a 2.5 approximation for
correlation clustering at the preceding STOC'05 with
Newman \cite{DBLP:journals/jacm/AilonCN08}. In fact, inspired by this connection they proposed in \cite{DBLP:journals/siamcomp/AilonC11} a pivot-based $(M+2)$-approximation
algorithm for the $L_1$ ultrametric problem where $M$ is the number of distinct input distances.

\subsubsection{Hierarchical correlation clustering}\label{sec:HierCorrClust}
We are going to work with a generalization of
the problem of $L_1$-fitting ultrametric which is 
implicit in previous work \cite{DBLP:journals/siamcomp/AilonC11, 12ac-DBLP:conf/approx/HarbKM05}, but where we will exploit
the generality in new interesting ways. \smallskip\\
\noindent{\bf Problem} Hierarchical Correlation Clustering. \\
\noindent{\bf Input} The input is $\ell$ weights $\delta^{(1)},\ldots,\delta^{(\ell)}\in \mathbb R_{>0}$ and $\ell$
edge sets $E^{(1)},\ldots,E^{(\ell)}\subseteq{S\choose 2}$.\\
\noindent{\bf Desired output} $\ell$ partitions $P^{(1)},\ldots,P^{(\ell)}$ of $S$ that are 
hierarchical in the sense that $P^{(t)}$ subdivides $P^{(t+1)}$ and such that we minimize
\begin{equation}\label{eq:hier-corr}
\sum_{t=1}^\ell \delta^{(t)} |E^{(t)}\, \Delta\, \calE(P^{(t)})|
\end{equation}
Thus we are having a combination of $\ell$ correlation
clustering problems where we want the output partitions to form a hierarchy, and where the objective is
to minimize a weighted sum of the costs for
each level problem.

We shall review the reduction from $L_1$-fitting of
ultrametrics to hierarchical correlation clustering in Section \ref{sec:UltraMetric-HierCorrClust}. The instances we get from ultrametrics will always
satisfy $E^{(1)}\subseteq\cdots\subseteq E^{(\ell)}$, but as we shall see shortly, our new algorithms will reduce to instances where this is not the case, even if the original input is from
an ultrametric.

\subsubsection{Hierarchical cluster agreement} \label{sec:HierCorrAgree}
We will be particularly interested in the following
special case of \GeneralProblem.

\noindent{\bf Problem} Hierarchical Cluster Agreement. \\
\noindent{\bf Input} The input is $\ell$ weights $\delta^{(1)},\ldots,\delta^{(\ell)}\in \mathbb R_{>0}$ and $\ell$
partitions $Q^{(1)},\ldots Q^{(\ell)}$ of $S$.\\
\noindent{\bf Desired output} $\ell$ partitions $P^{(1)},\ldots,P^{(\ell)}$ of $S$ that are 
hierarchical in the sense that $P^{(t)}$ subdivides $P^{(t+1)}$ and such that we minimize
\begin{equation}\label{eq:hier-clust}
\sum_{t=1}^\ell \delta^{(t)} |\calE(Q^{(t)})\, \Delta\, \calE(P^{(t)})|
\end{equation}

This is the special case of hierarchical correlation clustering, where the input edge set $E^{(t)}$
are the disjoint clique edges from $\calE(Q^{(t)})$. The challenge is that the input partitions may disagree in the sense that $Q^{(t)}$ does not subdivide
$Q^{(t+1)}$, or equivalently, $\calE(Q^{(t)})\not\subseteq\calE(Q^{(t+1)})$, so now
we have to find the best hierarchical agreement.

We are not aware of any previous work on hierarchical cluster agreement, but it plays a central role in our hierarchical correlation clustering algorithm, outlined below. 

\subsection{High-level algorithm for hierarchical correlation clustering}
Our main technical contribution in this paper is solving \GeneralProblem.
Reductions from Ultrametric to \GeneralProblem{}, and from general Tree Metric to Ultrametric are already known from \cite{DBLP:journals/siamcomp/AilonC11, 12ac-DBLP:conf/approx/HarbKM05} and \cite{DBLP:journals/siamcomp/AgarwalaBFPT99} respectively. We discuss both reductions in Sections~\ref{sec:UltraMetric-HierCorrClust} and \ref{sec:tree-3ultra}. This
includes removing some restrictions in the reduction
from Tree Metrics to Ultrametrics.

Focusing on \GeneralProblem, our input is the $\ell$ weights $\delta^{(1)},\ldots,\delta^{(\ell)}\in \mathbb R_{>0}$ and $\ell$
edge sets $E^{(1)},\ldots,E^{(\ell)}\subseteq{S\choose 2}$.

\paragraph{Step 1: Solve correlation clustering independently for each level}
The first step in our solution is to solve the correlation clustering problem defined by $E^{(t)}$
for each level $t=1,\ldots,\ell$ independently, 
thus obtaining an intermediate partitioning
$Q_t$. As we mentioned in Section \ref{sec:CorrClust},
this can be done so that $Q_t$ minimizes $|E^{(t)}\Delta E(Q_t)|$ within a constant
factor. 

\paragraph{Step 2: Solve hierarchical cluster agreement}

We now use the $\ell$ weights $\delta^{(1)},\ldots,\delta^{(\ell)}\in \mathbb R_{>0}$ and $\ell$
partitions $Q^{(1)},\ldots,Q^{(\ell)}$ of $S$ as input to the hierarchical cluster agreement problem, which we 
solve using an LP very similar to
the one Ailon and Charikar \cite{DBLP:journals/siamcomp/AilonC11} used to
solve general hierarchical correlation clustering.
However, when applied to the special case of hierarchical cluster agreement, it allows a special simple LP rounding where the LP decides which sets from the input partitions are important to the hierarchy, and which sets can be ignored. Having decided the important sets, it turns out that a very simple combinatorial algorithm can generate 
the hierarchical output partitions $P^{(1)},\ldots, P^{(\ell)}$ bottom-up. The
result is a poly-time constant factor approximation
for hierarchical cluster agreement, that is
\begin{equation*}
    \text{\HCA}=\calO(1)
\end{equation*}
The output partitions $P^{(1)},\ldots, P^{(\ell)}$ are
also returned as output to the original hierarchical
correlation clustering problem.

We now provide a high level overview and the intuition behind the hierarchical cluster agreement algorithm. The algorithm can be broadly divided into two parts.

\textbf{LP cleaning.}
We start by optimally solving the LP based on the weights $\delta^{(1)},\dots,\delta^{(\ell)}$ and partitions $Q^{(1)},\dots,Q^{(\ell)}$.
For each level $t\in \{1,\ldots, \ell\}$, we naturally think of the relevant LP variables as distances, and call them LP distances.
That is because a small value means that the LP wants the corresponding species to be in the same part of the output partition at level $t$, and vice versa, while the LP constraints also enforce the triangle inequality.
The weights $\delta^{(1)},\dots,\delta^{(\ell)}$
impact the optimal LP variables, but will otherwise not be used in the rest of the algorithm.

Using the LP distances, we clean each set in every partition $Q^{(t)}$ independently.
The objective of this step (LP-Cleaning - Algorithm~\ref{algo:lp-cleaning}) is to keep only the sets whose species are very close to each other and far away from the species not in the set. All other sets are disregarded.
Even though this is not a binary decision, it can be thought of as one, since the algorithm may only slightly modify each surviving set.
The property that we can clean each set independently to decide whether it is important or not, without looking at any other sets makes this part of our algorithm quite simple. 

Omitting exact thresholds for simplicity, the algorithm works as follows. 
We process each set $C_I\in Q^{(t)}$ by keeping only those species that are at very small LP distance from at least half of the other species in $C_I$ and at large LP distance to almost all the species outside $C_I$. Let us note that by triangle inequality and the pigeonhole principle, all species left in a set are at relatively small distance from each other.
After this cleaning process, we only keep a set if at least $90\%$ of its species are still intact, and we completely disregard it otherwise. 
The LP cleaning algorithm outputs the sequence $L^{(*)}=(L^{(1)},\dots,L^{(\ell)})$ where $L^{(t)}$ is the family
of surviving cleaned sets from $Q^{(t)}$.

\textbf{Derive hierarchy.} Taking $L^{(*)}$ as input, in the next step the algorithm Derive-Hierarchy (Algorithm~\ref{algo:Derive-hierarchy}) computes a hierarchical partition $P^{(*)}=(P^{(1)},\dots,P^{(\ell)})$ of $S$. 
This algorithm works bottom-up while initializing an auxiliary bottom most level of the hierarchy with $|S|$ sets where each set is a singleton and corresponds to a species of $S$.
Then the algorithm performs $\ell$ iterations where at the $t$-th iteration it processes all the disjoint sets in $L^{(t)}$ and computes  partition $P^{(t)}$ while ensuring that at the end of the iteration $P^{(1)},\dots,P^{(t)}$ are hierarchical.
An interesting feature of our algorithm is that, once created, no further computation processing the upper levels can modify the already created partitions.
Next, we discuss how to compute $P^{(t)}$ given $L^{(t)}$ and all the lower level sets in partitions $P^{(1)},\dots,P^{(t-1)}$.

Consider a set $C_{LP}\in L^{(t)}$.
Now if for each lower level set $C'$, either $C'\cap C_{LP}=\emptyset$ or $C'\subseteq C_{LP}$, then introducing $C_{LP}$ at level $t$ does not violate the hierarchy property.
Otherwise let $C'$ be a lower level set such that $C'\cap C_{LP}\neq \emptyset$ and $C'\not \subseteq C_{LP}$.
Note that we already mentioned, once created, $C'$ is never modified while processing upper level sets.
Thus, to ensure the hierarchy condition, the algorithm can either extend $C_{LP}$ so that it completely covers $C'$ or can discard the common part from $C_{LP}$.

In the process of modifying $C_{LP}$ (where we can add or discard some species from it), at any point we define the core of $C_{LP}$ to be the part that comes from the initial set.
Now to resolve the conflict between $C_{LP}$ and $C'$ we work as follows.
If the core of $C_{LP}$ intersects the core of $C'$ then we extend $C_{LP}$ so that $C'$ becomes a subset of it.
Omitting technical details, there are two main ideas here: first, we ensure that the number of species in $C'$ (resp. $C_{LP})$ that are not part of its core is negligible with respect to the size of $C'$ (resp. $C_{LP})$. Furthermore, since the cores of $C_{LP}, C'$ have at least one common species, using triangular inequality we can claim that any pair of species from the cores of $C', C_{LP}$ also have small LP distance; therefore, nearly all pairs of species in $C_{LP}, C'$ have small LP distance, meaning that the extension of $C_{LP}$ is desirable (i.e. its cost is within a constant factor from the LP cost while it ensures the hierarchy).

Here we want to emphasize the point that because of the LP-cleaning, we can ensure that for any lower level set $C'$ at level $t$ there exists at most one set whose core has an intersection with the core of $C'$. We call this the \emph{hierarchy-friendly} property of the LP cleaned sets.
This property is crucial for consistency, as it ensures that at level $t$ there cannot exist more than one sets that are allowed to contain $C'$ as a subset.

In the other case, where the cores of $C_{LP}$ and $C'$ do not intersect, the algorithm removes $C_{LP}\cap C'$ from $C_{LP}$. The analysis of this part is more technical but follows the same arguments, namely using the aforementioned properties of LP-cleaning along with triangle inequality.

After processing all the sets in $L^{(t)}$, the algorithm naturally combines these processed sets with $P^{(t-1)}$ to generate $P^{(t)}$, thus ensuring that $P^{(1)}, \ldots, P^{(t)}$ are hierarchical.

%%%%%%%%%%%%%%%%

\paragraph{High-level analysis}
We will prove that the partitions $P^{(1)},\ldots, P^{(\ell)}$ solves the original hierarchical clustering problem within a constant factor.

Using triangle inequality, we are going to show that the switch in Step 1,
from the input edge sets $E^{(1)},\ldots,E^{(\ell)}$ to the
partitions $Q^{(1)},\ldots,Q^{(\ell)}$ costs us
no more than the approximation factor
of correlation clustering used to generate each partition. This then becomes a multiplicative factor on our approximation factor for hierarchical cluster agreement, more specifically,
\begin{align*}
        \text{\HCC}     &< (\text{\CC}+1) (\text{\HCA}+1)
\end{align*}
We will even show that we can work with the LP from
\cite{DBLP:journals/siamcomp/AilonC11} for the original hierarchical correlation clustering problem, and get
a solution demonstrating a constant factor integrality
gap.

%\subsection{Further Related Work}
%\input{further_related_work}

\subsection{Organization of the paper}
In Section~\ref{sec:lp-defs} we present the LP formulation and related definitions for \GeneralProblem{}. In Section~\ref{sec:hcc_to_hca} we show how to reduce \GeneralProblem{} to \NewProblem{}. In Section~\ref{sec:Hierarchical_Cluster_Agreement} we present the algorithm for \NewProblem{}, and in Section~\ref{sec:analysis_HCA} we analyze it. In Section~\ref{sec:integralityGap} we show that the LP formulation for \GeneralProblem{} has constant integrality gap. In Section~\ref{sec:UltraMetric-HierCorrClust} we show how $L_p$-fitting ultrametrics reduces to \GeneralProblem{}. In Section~\ref{sec:tree-3ultra} we discuss the reduction from $L_p$-fitting tree metrics to $L_p$-fitting ultrametrics. In Section~\ref{sec:apx} we prove APX-Hardness of $L_1$-fitting ultrametrics and $L_1$-fitting tree metrics. We conclude in Section~\ref{sec:conclusions}.

\section{LP definitions for \GeneralProblem}\label{sec:lp-defs}
In this section we present the IP/LP formulation of \GeneralProblem, implicit in \cite{DBLP:journals/siamcomp/AilonC11, 12ac-DBLP:conf/approx/HarbKM05}. In what follows we use $[n]$ to denote the set $\{1, \ldots, n\}$.

\begin{definition}[IP/LP formulation of \GeneralProblem]
Given is a set $S$, $\ell$ positive numbers $\delta^{(1)}, \ldots, \delta^{(\ell)}$ and edge-sets $E^{(1)}, \ldots, E^{(\ell)} \subseteq {S \choose 2}$. The objective is:

$$min \sum_{t=1}^{\ell}\delta^{(t)} \left( \sum_{\{i,j\} \in E^{(t)}}x_{i,j}^{(t)} + \sum_{\{i,j\} \not \in E^{(t)}}(1-x_{i,j}^{(t)})\right)$$
subject to the constraints
\begin{align}
	x_{i,j}^{(t)} &\le x_{i,k}^{(t)} + x_{j,k}^{(t)} &\forall \{i,j,k\} \in {S \choose 3}, t\in [\ell] \label{ineq:triangle}\\
	x_{i,j}^{(t)} &\ge x_{i,j}^{(t+1)} &\forall \{i,j\} \in {S \choose 2}, t\in [\ell-1] \label{ineq:hier}\\
	x_{i,j}^{(t)} &\in \left\{
\begin{array}{ll}
      ~\{0,1\} &\text{if~IP}\\
      ~[0,1]   &\text{if~LP}\\
\end{array} \right. &\forall \{i,j\} \in {S \choose 2}, t\in [\ell]
\end{align}
\end{definition} 

Concerning the IP, the values $x_{i,j}^{(t)}$ encode the hierarchical partitions, with $x_{i,j}^{(t)}=0$ meaning that $i,j$ are in the same part of the partition at level $t$, and $x_{i,j}^{(t)}=1$ meaning that they are not. 
Inequality~(\ref{ineq:triangle}) ensures that the property of being in the same part of a partition is transitive.
Inequality~(\ref{ineq:hier}) ensures that the partitions are hierarchical.

In the LP, where fractional values are allowed, $x_{i,j}^{(t)}$ is said to be the {\em LP-distance} between $i,j$ at level $t$. If their LP-distance is small, one should think of it as the LP suggesting that $i,j$ should be in the same part of the output partition, while a large LP-distance suggests that they should not. Notice that for any given level $t$, the LP-distances satisfy the triangle inequality, by (\ref{ineq:triangle}).

We also note that the Correlation Clustering problem directly corresponds to the case where $\ell=\delta_1=1$. 

\section{From \GeneralProblem~to \NewProblem~Problem}
\label{sec:hcc_to_hca}
Our main technical contribution is proving the following theorem.

\begin{theorem}\label{thm:hcc}
The \GeneralProblem{} problem can be solved in deterministic polynomial time within a constant approximation factor.
\end{theorem}

In this section, we present a deterministic reduction from \GeneralProblem{} to \NewProblem{} that guarantees:

\begin{align}
        \text{\HCC}     &< (\text{\CC}+1) (\text{\HCA}+1)
    \tag*{\req{eq:reduction_HCC_to_HCA} from Figure~\ref{fig:roadmap}}
\end{align}

In Sections~\ref{sec:Hierarchical_Cluster_Agreement} and \ref{sec:analysis_HCA} we present a deterministic polynomial time constant factor approximation algorithm for \NewProblem; combined with a known deterministic polynomial time constant factor approximation algorithm for Correlation Clustering \cite{DBLP:journals/mor/ZuylenW09}, it completes the proof of Theorem~\ref{thm:hcc}.

Assume that Correlation Clustering can be approximated within a factor $\alpha$ and that Hierarchical Cluster Agreement can be approximated within a factor $\beta$ (Section~\ref{sec:Hierarchical_Cluster_Agreement}). We prove Inequality~(\ref{eq:reduction_HCC_to_HCA}), by providing an algorithm to approximate Hierarchical Correlation Clustering within a factor $(\alpha+1)(\beta+1)-1$. 

Suppose we have a Hierarchical Correlation Clustering instance
$S, \delta^{(1)}, \ldots, \delta^{(\ell)}, E^{(1)}, \ldots, E^{(\ell)}$.
Our algorithm first solves the Correlation Clustering instance $S,E^{(t)}$ to acquire partition $Q^{(t)}$, for every level $t$.
Then, we solve the \NewProblem~instance
$S, \delta^{(1)}, \ldots, \delta^{(\ell)}, \calE(Q^{(1)}), \ldots,$ $\calE(Q^{(\ell)})$
and obtain hierarchical partitions
$P^{(1)}, \ldots, P^{(\ell)}$.

The proof that the hierarchical partitions $P^{(1)}, \ldots, P^{(\ell)}$ are a good approximation to the Hierarchical Correlation Clustering instance follows from two observations. First, by definition, the cost of Hierarchical Correlation Clustering is related to certain symmetric differences. Since the cardinality of symmetric differences satisfy the triangle inequality, we can switch between the cost of Hierarchical Correlation Clustering and Hierarchical Cluster Agreement under the same output, with only an additive term related to $|E^{(t)} \triangle \calE(Q^{(t)})|$ and not related to the output. Second, by definition of $Q^{(t)}$, the cardinality of the symmetric difference $|E^{(t)} \triangle \calE(Q^{(t)})|$ is minimized within a factor $\alpha$.

More formally, for this proof we need to define the following three concepts:

\begin{itemize}
    \item For any $t\in [\ell]$, $OPT_{\CC}^{(t)}$ is an optimal solution to the Correlation Clustering instance at level $t$, that is a partition minimizing 
    $$|E^{(t)} \triangle \calE(OPT_{\CC}^{(t)})|$$
    \item $OPT_{\HCC}=(OPT_{\HCC}^{(1)},\ldots,OPT_{\HCC}^{(\ell)})$ is an optimal solution to the \GeneralProblem~instance, that is a sequence of hierarchical partitions minimizing
    $$\sum_{t=1}^{\ell} \delta^{(t)} |E^{(t)} \triangle \calE(OPT_{\HCC}^{(t)})|$$
    \item $OPT_{\HCA}=(OPT_{\HCA}^{(1)},\ldots,OPT_{\HCA}^{(\ell)})$ is an optimal solution to the \NewProblem~instance, that is a sequence of hierarchical partitions minimizing
    $$\sum_{t=1}^{\ell} \delta^{(t)} |\calE(Q^{(t)}) \triangle \calE(OPT_{\HCA}^{(t)})|$$
\end{itemize}
Notice, for any $t$, the difference between $OPT_{\CC}^{(t)}$ and $OPT_{\HCC}^{(t)}$. The first one optimizes locally (per level), meaning that $|E^{(t)} \triangle \calE(OPT_{\CC}^{(t)})| \le |E^{(t)} \triangle \calE(OPT_{\HCC}^{(t)})|$, and therefore $\sum_{t=1}^{\ell} \delta^{(t)} |E^{(t)} \triangle \calE(OPT_{\CC}^{(t)})| \le \sum_{t=1}^{\ell} \delta^{(t)} |E^{(t)} \triangle \calE(OPT_{\HCC}^{(t)})|$. This does not contradict the definition of $OPT_{\HCC}$, as the sequence $OPT_{\CC}^{(1)}, \ldots, OPT_{\CC}^{(\ell)}$ is not a sequence of hierarchical partitions.

The cost of our solution is 
\begin{align}
  \sum_{t=1}^{\ell} \delta^{(t)} |E^{(t)} \triangle \calE(P^{(t)})|
  \leq
  \sum_{t=1}^{\ell} \delta^{(t)} |E^{(t)} \triangle \calE(Q^{(t)})|
+
  \sum_{t=1}^{\ell} \delta^{(t)} |\calE(Q^{(t)}) \triangle \calE(P^{(t)})|
  \label{ineq:reduction_HCC_to_HCA}
\end{align}

By definition of $P^{(1)}, \ldots, P^{(\ell)}$, they minimize the second term of $(\ref{ineq:reduction_HCC_to_HCA})$ within a factor $\beta$. Therefore, the second term is upper bounded by $\beta\sum_{t=1}^{\ell} \delta^{(t)} |\calE(Q^{(t)}) \triangle \calE(OPT_{\HCA}^{(t)})|$, which, by optimality of $OPT_{\HCA}$ is upper bounded by $\beta\sum_{t=1}^{\ell} \delta^{(t)} |\calE(Q^{(t)}) \triangle \calE(OPT_{\HCC}^{(t)})|$.

Using the triangle inequality again, we further upper bound the second term by:

$$\beta \sum_{t=1}^{\ell} \delta^{(t)} |\calE(Q^{(t)}) \triangle E^{(t)}| + \beta\sum_{t=1}^{\ell} \delta^{(t)} |E^{(t)} \triangle \calE(OPT_{\HCC}^{(t)})|$$

Therefore, we can rewrite $(\ref{ineq:reduction_HCC_to_HCA})$ as:

$$ \sum_{t=1}^{\ell} \delta^{(t)} |E^{(t)} \triangle \calE(P^{(t)})|
  \leq
  (\beta+1) \sum_{t=1}^{\ell} \delta^{(t)} |E^{(t)} \triangle \calE(Q^{(t)})| + \beta\sum_{t=1}^{\ell} \delta^{(t)} |E^{(t)} \triangle \calE(OPT_{\HCC}^{(t)})|$$

Since $Q^{(t)}$ is obtained by solving Correlation Clustering at level $t$ within a factor $\alpha$, we get
$$\sum_{t=1}^{\ell} \delta^{(t)}|E^{(t)} \triangle \calE(Q^{(t)})| \le \alpha\sum_{t=1}^{\ell} \delta^{(t)} |E^{(t)} \triangle \calE(OPT_{\CC}^{(t)})|$$

By optimality of $OPT_{\CC}^{(t)}$, for each $t\in [\ell]$, we have

$$\sum_{t=1}^{\ell} \delta^{(t)} |E^{(t)} \triangle \calE(OPT_{\CC}^{(t)})| \le \sum_{t=1}^{\ell} \delta^{(t)} |E^{(t)} \triangle \calE(OPT_{\HCC}^{(t)})|$$

which proves that

\begin{align*}
\sum_{t=1}^{\ell} \delta^{(t)} |E^{(t)} \triangle \calE(P^{(t)})|
  &\leq \left((\beta+1)\alpha + \beta \right) \sum_{t=1}^{\ell} \delta^{(t)} |E^{(t)} \triangle \calE(OPT_{\HCC}^{(t)})|\\
  &= \left((\alpha+1)(\beta+1)-1\right) \sum_{t=1}^{\ell} \delta^{(t)} |E^{(t)} \triangle \calE(OPT_{\HCC}^{(t)})|
\end{align*}

\section{Constant approximation Algorithm for Hierarchical Cluster \\Agreement}
\label{sec:Hierarchical_Cluster_Agreement}
In this section we introduce our main algorithm, which consists of three parts: Solving the LP formulation of the problem, the LP-Cleaning subroutine and the Derive-Hierarchy subroutine.

Informally, the LP-Cleaning subroutine uses the fractional solution of the LP relaxation of \NewProblem~to decide which of our input-sets are important and which are not.
The decision is not a binary one, because important sets are also cleaned, in the sense that bad parts of them may be removed.
However, at least a $\deltaInLp$ fraction of them is left intact, while unimportant sets are completely discarded.

The Derive-Hierarchy part then receives the cleaned input-sets by LP-Cleaning, and applies a very simple combinatorial algorithm on them to compute the output.

We notice that the weights $\delta^{(*)}$ are only used for solving the LP. Moreover, the fractional LP-solution is only used by LP-Cleaning to guide this ``nearly-binary'' decision for each input-set. The rest of the algorithm is combinatorial and does not take the LP-solution into account.

\subsection{LP Definitions for \NewProblem} \label{ssec:lp-definitions-hca}
The IP-formulation of \NewProblem~is akin to the IP-formulation of \GeneralProblem. Namely, the constraints are exactly the same for both problems. The only difference is in the objective function where we replace the general edge-sets $E^{(1)}, \ldots, E^{(\ell)}$ with the disjoint clique edges from $\calE(Q^{(1)}), \ldots, \calE(Q^{(\ell)})$. Similarly for the LP-relaxation of \NewProblem. Here each component in $Q^{(t)}$ is called a \emph{level}-$t$ \emph{input cluster}.

To simplify our discussion, we use $\FOPT$ to denote a fractional solution to the LP-relaxation of \NewProblem, that is a vector containing all $x_{i,j}^{(t)}, \{i,j\}\in {S\choose 2}, t\in [\ell]$. One can think of $\FOPT$ as the optimal fractional solution, but in principle it can be any solution.

We use $x^{(t)}$, for some particular $t\in [\ell]$, to denote the vector containing all $x_{i,j}^{(t)}, \{i,j\}\in {S\choose 2}$.

As previously, we use the term LP distances to refer to the entries of $x^{(*)}$, and notice that for any particular $t\in [l]$ the LP distances even satisfy the triangle inequality, by the LP constraints.

Given $\FOPT$ we define $\B{r}{t}(i)$ to be the ball of species with LP-distance less than $r$ from $i$ at level $t$. 
More formally, $\B{r}{t}(i) = \{j\in S \mid x_{i,j}^{(t)}<r\}$. Similarly, for a subset $S'$ of $S$ we define the ball $\B{r}{t}(S') = \{j\in S \mid \exists i\in S' \text{ s.t. }x_{i,j}^{(t)}<r\}$.

We also define the LP cost of species $i,j$ at level $t$ as 
$$cost_{i,j}^{(t)} = \left\{
\begin{array}{ll}
      ~\delta^{(t)}x_{i,j}^{(t)}    ~~~~~~~~~~\text{if } \{i,j\}\in \calE(Q^{(t)})\\
      ~\delta^{(t)}(1-x_{i,j}^{(t)})       ~~~\text{otherwise}\\
\end{array} \right.$$

as well as the LP cost of species in a set $S'\subseteq S$ at level $t$ as
$$cost_{S'}^{(t)} = \sum_{\substack{\{i,j\} \in {S\choose 2}\\i\in S' \text{ or } j\in S'}} cost_{i,j}^{(t)}$$

and in case $S'$ only contains a single species $i$, we write $cost_i^{(t)}$ instead of $cost_{\{i\}}^{(t)}$.

Then the LP cost at level $t$ is denoted as $cost^{(t)} = cost_{S}^{(t)}$.

Finally, the LP cost is simply $cost^{(*)} = \sum_{t=1}^{\ell} cost^{(t)}$.

\subsection{Main Algorithm}
The pseudocode for our main algorithm for \NewProblem~is given in Algorithm~\ref{algo:main}.

\begin{algorithm}[H]
\caption{\NewProblem Algorithm}
\label{algo:main}
\hspace*{\algorithmicindent} \textbf{Input} \hspace*{0.65 \dimexpr \algorithmicindent} A set $S$, a sequence $Q^{(*)}=(Q^{(1)}, \cdots, Q^{(\ell)})$ of partitions of $S$, and weights\\
\hspace*{4 \dimexpr \algorithmicindent} $\delta^{(*)}=(\delta^{(1)}, \cdots, \delta^{(\ell)})$\\
\hspace*{\algorithmicindent} \textbf{Returns} A sequence $P^{(*)}=(P^{(1)}, \cdots, P^{(\ell)})$ of hierarchical partitions of $S$
\label{algo:main-algorithm}
\begin{algorithmic}[1]
\State $x^{(*)}\gets$ Solve(LP-relaxation($S,Q^{(*)},\delta^{(*)}$))
\State $L^{(*)} \gets \text{LP-Cleaning}(S, Q^{(*)}, x^{(*)})$
\State \Return $\text{Derive-Hierarchy}(S, L^{(*)})$
\end{algorithmic}
\end{algorithm}

Our LP relaxation has size polynomial in $S, \ell$, and the two subroutines also run in polynomial time, as we show later. Therefore the whole algorithm runs in polynomial time.

\subsection{LP cleaning Algorithm}
In Algorithm~\ref{algo:lp-cleaning} we provide the pseudocode of the LP Cleaning step of our algorithm.

Intuitively, the aim of this algorithm is to clean the input sets so that (ideally) all species remaining in a set have small LP distances to each other, and large LP distances to species not in the set.

\begin{algorithm}[H]
\caption{LP-Cleaning}
\hspace*{\algorithmicindent} \textbf{Input} \hspace*{0.65 \dimexpr \algorithmicindent} A set $S$, a sequence $Q^{(*)}=(Q^{(1)}, \cdots, Q^{(\ell)})$ of partitions of $S$,\\
\hspace*{4 \dimexpr \algorithmicindent}
and a fractional solution $x^{(*)}$\\
\hspace*{\algorithmicindent} \textbf{Returns} A sequence $L^{(*)}=(L^{(1)}, \cdots, L^{(\ell)})$ of families of disjoint subsets of $S$
\label{algo:lp-cleaning}
\begin{algorithmic}[1]
\For {$t\gets 1, \ldots, \ell$} 
    \State $L^{(t)} \gets \emptyset$
	\For {$C_I \in Q^{(t)}$} \label{line:inputCluster}
		\State $C_{LP} \gets \left\{i\in C_I \left| \begin{array}{l}
|\B{\radiusInCluster}{t}(i) \cap C_I| > \frac12|C_I|, \\[.5ex] |\B{\threeDiameter}{t}(i) \setminus C_I| \le \badOutEdges|C_I| \end{array} \right. \right\}$ \label{line:outsideA}
        \If{$|C_{LP}| \ge \deltaInLp |C_I|$} \label{line:deltaIL}
		    \State $L^{(t)} \gets L^{(t)} \cup \{C_{LP}\}$
		\EndIf
	\EndFor
\EndFor
\State \Return $L^{(*)}=(L^{(1)}, \cdots, L^{(\ell)})$
\end{algorithmic}
\end{algorithm}

%\todo[inline]{E: Added this description of the algorithm in case people avoid to look at the pseudocode (I know I usually do).}
%In short, we first solve the LP relaxation to acquire all the LP distances. Then, for each input partition $Q^{(t)}$, we process all its parts. For each one, we remove the species that either do not have small LP distance to more than half the species in the part, or have small LP distance to many species not in the part (more than \badOutEdges~times its initial size). We completely discard the ones that have less than \deltaInLp~of their initial size remaining. Finally, we return a family of (remaining) sets for each input partition.

%Referring to Algorithm~\ref{algo:lp-cleaning}, we use the term LP-cluster for any set that was not discarded, that is for any $C_{LP}$ of the algorithm such that $|C_{LP}|\ge \deltaInLp |C_I|$. We use the term input-cluster for the corresponding $C_I$. The level of a cluster is defined by the value of $t$ when the algorithm processes this cluster.

%Out of the many things we later prove concerning the output of LP Cleaning, we briefly discuss the following property: No two LP-clusters at the same level $t$ may be intersected by the same LP-cluster at level $t'<t$. We call this property of the output of LP Cleaning the {\em hierarchy-friendly} property.

Formally Algorithm~\ref{algo:lp-cleaning} takes a sequence $Q^{(*)}=(Q^{(1)},\dots,Q^{(\ell)})$ of partitions of $S$ and a fractional solution $x^{(*)}$ containing LP distances. It outputs a sequence of families of disjoint subsets of $S$, $L^{(*)}=(L^{(1)},\dots,L^{(\ell)})$. Here each component of $L^{(t)}$ is called a \emph{level}-$t$ \emph{LP-cluster}. 

In the algorithm, for each input partition $Q^{(t)}$ we process every level-$t$ input-cluster $C_I\in Q^{(t)}$ separately. For this we remove all the species in $C_I$ that do not have very small LP distance to at least half the species in $C_I$ or that have small LP distance to many species not in $C_I$. More formally, we remove all the species in $C_I$ with LP distance less than $0.1$ to at most half the species in $C_I$ or with LP distance less than $0.6$ to more than $0.05|C_I|$ species not in $C_I$.

After the cleaning step we discard $C_I$ if less than $9/10$ fraction of the species survive. Otherwise we create an LP-cluster $C_{LP}$ containing the species in $C_I$ that survive. Next we add the level-$t$ LP-cluster $C_{LP}$ to $L^{(t)}$.   

Out of several properties that we prove concerning the output of the LP-Cleaning, we briefly mention the following one: The output sequence $L^{(*)}$ is {\em hierarchy-friendly} in the sense that no two LP-clusters at the same level $t$ can be intersected by the same LP-cluster at level $t'<t$. We formally prove this in Lemma~\ref{lem:oneIntersection}.

The LP-Cleaning subroutine trivially runs in time polynomial in $S,\ell$.

\subsection{Derive-hierarchy Algorithm}
In this section, we introduce Derive-Hierarchy (Algorithm~\ref{algo:Derive-hierarchy}). It takes as input a hierarchy-friendly sequence $L^{(*)}=(L^{(1)}, \cdots, L^{(\ell)})$ of families of disjoint subsets of $S$  and outputs a sequence ${P^{(*)}=(P^{(1)}, \cdots, P^{(\ell)})}$ of hierarchical partitions of $S$. The execution of the algorithm can be seen, via a graphical example, in Figure~\ref{fig:DeriveHierarchy}.

\begin{algorithm}
\caption{Derive-Hierarchy}
\hspace*{\algorithmicindent} \textbf{Input} \hspace*{0.65 \dimexpr \algorithmicindent} A set $S$, and a hierarchy-friendly sequence $L^{(*)}=(L^{(1)}, \cdots, L^{(\ell)})$\\
\hspace*{4 \dimexpr \algorithmicindent} of families of disjoint subsets of $S$\\
\hspace*{\algorithmicindent} \textbf{Returns} A sequence $P^{(*)}=(P^{(1)}, \cdots, P^{(\ell)})$ of hierarchical partitions of $S$
\label{algo:Derive-hierarchy}
\begin{algorithmic}[1]
\State Construct an empty forest $\calF$

\For {$i\in S$}
\State Create a singleton tree $T$ with a node $u_i$ and add it to $\calF$
\State Set $C(u_i) \gets C^+(u_i) \gets \{i\}$
\EndFor 
\For {$t\gets 1, \ldots, l$} 
	\For {$C_{LP} \in L^{(t)}$\label{line:main-iter}}
	    \State Create a node $u$ and set $C(u) \gets C_{LP}$ \label{line:creationNode}
		\For {all roots $v\in \calF$ s.t. $C(v) \cap C(u) = \emptyset$} \label{line:conditionRemoval}
		        \State $C(u) \gets C(u) \setminus C^+(v)$
		\EndFor
		\State $C^+(u) \gets C(u)$
		\For {all roots $v\in \calF$ s.t. $C(v) \cap C(u) \not= \emptyset$} \label{line:conditionExtension}
		        \State $C^+(u) \gets C^+(u) \cup C^+(v)$
		        \State Make $v$ a child of $u$ in $\calF$
%		        \State $w(v,u)\gets (\delta^{(t)}-\delta^{(level(v))})/2$
		\EndFor
	\EndFor
	\State Set $P^{(t)}$ to contain the extended-clusters $C^+(v)$ of all roots $v    \in \calF$
\EndFor
%\State \Return $\calF$
\State \Return $P^{(*)}=(P^{(1)}, \cdots, P^{(\ell)})$
\end{algorithmic}
\end{algorithm}

\begin{figure}
    \centering
    \includegraphics[width=190pt, height=125pt]{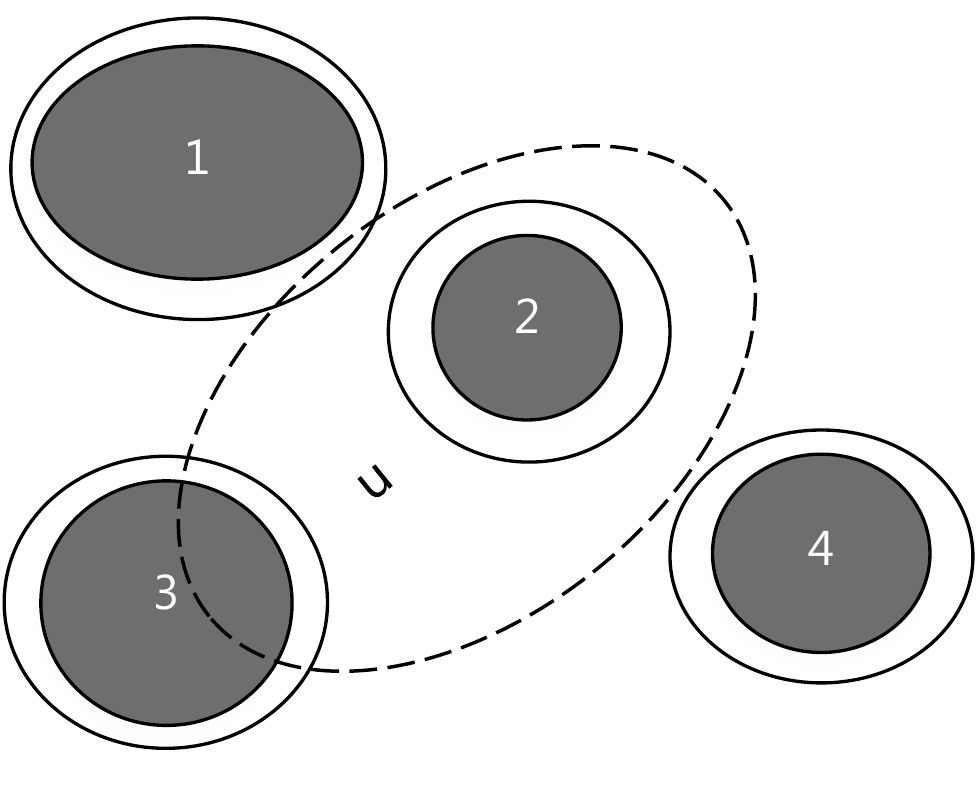}
    \hspace{40pt}
    \includegraphics[width=190pt, height=125pt]{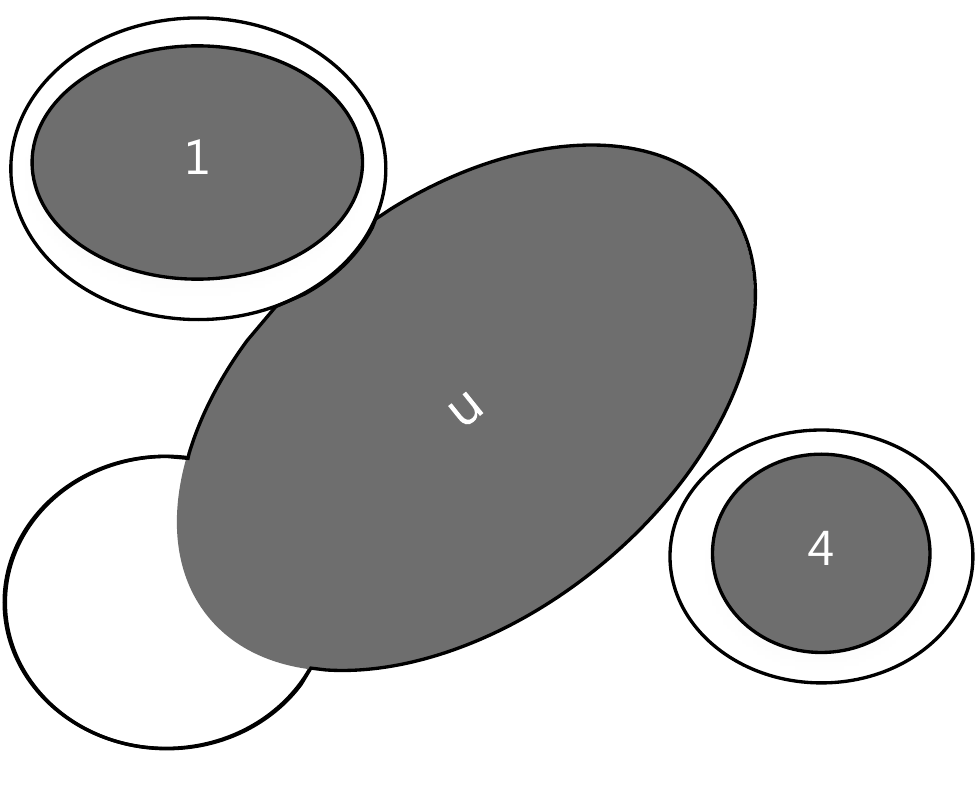}
    \caption{Example of Derive-Hierarchy (Algorithm~\ref{algo:Derive-hierarchy}). Nodes $1,2,3,4$ (left) are the roots of the forest before inserting the new LP-cluster $L(u)$ (dashed line). Each node is described by its extended-cluster, with the shaded part being the core-cluster. Core-cluster of nodes $2$ and $3$ intersect $L(u)$; thus they become children of $u$ and the extended-cluster of $u$ covers the extended-clusters of $2,3$ (right). Notice that the core-cluster of $u$ is reduced due to node $1$.}
    \label{fig:DeriveHierarchy}
\end{figure}

%\todo[inline]{E: Added this description of the algorithm in case people (like me) avoid to look at the pseudocode.}

The algorithm works bottom-up while performing $\ell$ iterations for $t=1,\dots,\ell$. In the process it incrementally builds a forest $\mathcal{F}$. 
Throughout the algorithm each non leaf node $u$ in $\mathcal{F}$ can be identified by an LP-cluster in $L^{(*)}$. Moreover for each node $u$ the algorithm maintains two sets $C(u)$ and $C^{+}(u) \subseteq S$. 

The algorithm starts by initializing $\mathcal{F}$ with $|S|$ trees where each tree contains a single node $u_i$ identified by a species $i\in S$. 
Also it initializes both sets $C(u_i),C^{+}(u_i)$ with $\{i\}$. 
%where $C(u)$ is a subset of the LP-cluster associated with node $u$ and $C^{+}(u)$ is the set of leaf nodes of the subtree rooted at node $u$.
Next in iteration $t$ the algorithm processes the LP-clusters in $L^{(t)}$ and at the end of the iteration, the $C^{+}()$ sets associated with the root nodes in $\mathcal{F}$ define the partition $P^{(t)}$. Precisely here, the $C^{+}()$ set of a root node contains all the species descending from the respective root.

In the $t$-th iteration, for each cluster $C_{LP}\in L^{(t)}$ the algorithm adds a root node $u$ to forest $\mathcal{F}$ while initializing the set $C(u)$ with $C_{LP}$. 
Next for the root node $u$ the algorithm decides on its children by processing the pre-existing roots in the following way.
For consistency first it detects all the pre-existing root nodes $v$ such that $C(u)$ does not intersect $C(v)$.
Then it removes from $C(u)$ all the species that are descending from $v$; i.e. sets $C(u)\gets C(u)\setminus C^{+}(v)$. Lastly it sets $u$ as a parent of all other pre-existing root nodes $v$ such that $C(u)$ intersects $C(v)$. Also accordingly it modifies the set of leaf nodes of the subtree rooted at $u$ by setting $C^+(u)\gets C^
+(u)\cup C^+(v)$. Notice here that some of the root-nodes may correspond to sets from levels lower than $t$, in case no parent was assigned to them.

At the end of iteration $t$ the algorithm completes processing all the LP-clusters in $L^{(1)},\dots,L^{(t)}$ and constructs partitions $P^{(1)},\dots,P^{(t)}$. At the end of the $\ell$ iterations it outputs the $\ell$ partitions $P^{(*)}=(P^{(1)},\dots,P^{(\ell)})$. 

The Derive-Hierarchy subroutine trivially runs in time polynomial in $S,\ell$.

%The algorithm works bottom-up. Informally, the approach is to incrementally build a forest where the leaves are $|S|$ nodes identified with the species in $S$. Suppose all sets in $L^{(1)}, \ldots, L^{(t-1)}$ have been processed. Then, for each set in $L^{(t)}$ we introduce an associated root node in the forest, which becomes the parent of certain pre-existing root nodes. After inserting all sets in $L^{(t)}$, we read the output partition $P^{(t)}$ by creating a set for each root node in the forest; this set contains all species descending from the respective root. Notice that some of the root-nodes may correspond to sets from levels lower than $t$, in case no parent was assigned to them. In the pseudocode, for any node $u$ we use $C^{+}(u)$ to refer to its descendant species.

%For a given new root node $u$ associated with some set $C(u)$, we decide on its children (pre-existing roots) in the following straightforward way. First, for consistency reasons, we detect the pre-existing root nodes whose associated set does not intersect $C(u)$; we remove all species associated with these root nodes from $C(u)$. Node $u$ then becomes a parent of all root nodes whose associated set intersects $C(u)$.

\section{Analysis of \NewProblem~Algorithm} \label{sec:analysis_HCA}
In this section, we proceed with our analysis. We first lay out some terminology, then provide some results related to the LP Cleaning, then some structural results, and finally prove that our algorithm is a constant factor approximation for \NewProblem.

\subsection{Terminology}
Notice that throughout the execution of the algorithm, $\calF$ is an incrementally updated graph (that is, no deletions occur). In fact it is always a forest, as we start with $|S|$ isolated nodes and only introduce new nodes as parents of roots of some of the existing trees. Moreover, this process implies that the subtree rooted at any specific node is never modified. 

From now on we use $\calF$ to refer to the final instance of the incrementally updated forest. We use $\calF(u)$ to refer to the state of this incrementally updated forest after introducing $u$; therefore $\calF(u)\setminus\{u\}$ denotes the state of the forest exactly before introducing node $u$. We naturally identify the leaves of $\calF$ with the species of $S$.

For any node $u$ in the forest $\calF$, the Derive-Hierarchy algorithm defines $C(u)$, which we call the core-cluster of $u$, and $C^+(u)$, which we call the extended-cluster of $u$. Furthermore, notice that each core-cluster $C(u)$ is a subset of some LP-cluster $C_{LP}$ (Line~\ref{line:creationNode} of Algorithm~\ref{algo:Derive-hierarchy}); we call this the LP-cluster of $u$ and denote it by $L(u)$. Moreover, each LP-cluster $L(u)$ is a subset of an input-cluster $C_I$ (Line~\ref{line:inputCluster} of Algorithm~\ref{algo:lp-cleaning}); we call this the input-cluster of $u$ and denote it by $I(u)$. These concepts are well defined for any new node $u$ and never change throughout the algorithm. We remind the reader that LP-Cleaning discards some of the input clusters, in the sense that they have no corresponding LP-cluster, and therefore they do not match $I(u)$, for any node $u$.

Directly from the algorithm we get that
\begin{alignat*}{2}
    C(u) &\subseteq L(u) \subseteq I(u) \\
    C(u) &\subseteq C^+(u)
\end{alignat*}

To help with our discussion, we also define the following variables related to the Derive-Hierarchy algorithm (Algorithm~\ref{algo:Derive-hierarchy}):
\begin{alignat*}{4}
    &\Delta^-(u) &&= L(u) &&\setminus C(u) \\
    & \Delta^{+}(u)  &&= C^+(u) &&\setminus C(u)
\end{alignat*}

For a node $u\in \calF$, we define its level $t(u)$ to be the value of iteration $t$ in Algorithm~\ref{algo:Derive-hierarchy} when internal node $u$ was introduced, and $0$ when $u$ is a leaf node.

\subsection{LP-Cleaning Results (Algorithm~\ref{algo:lp-cleaning})}
We start with some observations that are heavily used in proving structural results regarding the core and the extended-clusters. These are in turn used for lower-bounding the LP cost.

The most important reason we are using the LP-Cleaning subroutine is so that any two species belonging in the same LP-Cluster at level $t$ have small LP-distance.

\begin{lemma}
\label{lem:diameterCluster}
Given a node $u\in \calF$ and a species $i$ in $u$'s LP-cluster $L(u)$, it holds that the LP-distance from $i$ to any other species in $L(u)$ is less than $\diameterInCluster$ for all levels $t\ge t(u)$, that is $\B{\diameterInCluster}{t}(i) \supseteq L(u)$.
\end{lemma}
\begin{proof}
It suffices to prove that $x_{i,j}^{(t)} < \diameterInCluster$ for all $j\in L(u)$ only for level $t=t(u)$, as the LP constraints enforce $x_{i,j}^{(t+1)} \leq x_{i,j}^{(t)}$.

By pigeonhole principle, since both $\B{\radiusInCluster}{t}(i)\cap I(u)$ and $\B{\radiusInCluster}{t}(j) \cap I(u)$ have size more than $|I(u)|/2$ (Line~\ref{line:outsideA} of Algorithm~\ref{algo:lp-cleaning}), there exists a node $k \in I(u)$ for which both $x_{i,k}^{(t)}$ and $x_{j,k}^{(t)}$ are less than $\radiusInCluster$. 
Since the LP-distances in $x^{(t)}$ satisfy the triangle inequality, it follows that $x_{i,j}^{(t)} < \diameterInCluster$ (enforced by the LP constraints). 
\end{proof}

For the analysis, it is convenient that our relations involve the LP-clusters instead of the input-clusters. Therefore, we rephrase Line~\ref{line:outsideA} of Algorithm~\ref{algo:lp-cleaning} in terms of LP-clusters, effectively proving that few species outside of an LP-cluster $L(u)$ have small LP-distances to $L(u)$.

\begin{lemma}
\label{lem:outsideL}
For any node $u\in \calF$ it holds that $|\B{\twoDiameter}{t(u)}(L(u))| \le (1+\outsideL) |L(u)|$. In particular, $|\B{\twoDiameter}{t(u)}(L(u)) \setminus L(u)| \le \outsideL |L(u)|$.
\end{lemma}
\begin{proof}
Let $t=t(u)$.
We claim that species close to some species in $L(u)$ are close to all species in $L(u)$. Formally, we claim that for any $i\in L(u)$
$$\B{\twoDiameter}{t}(L(u)) \subseteq \B{\threeDiameter}{t}(i)$$

Let $j\in \B{\twoDiameter}{t}(L(u))$. We bound the LP-distance between $i,j$ by finding an intermediate $i'$ that is close to both and applying the triangle inequality forced by the LP constraints. By definition of $j$, there exists a species $i'\in L(u)$ with LP-distance less than $\twoDiameter$ from $j$. By Lemma~\ref{lem:diameterCluster}, the LP-distance between $i$ and $i'$ is less than $\diameterInCluster$, and thus by triangle inequality $x_{i,j}^{(t)}<\threeDiameter$.

Line~\ref{line:outsideA} of Algorithm~\ref{algo:lp-cleaning} gives that 
$$|\B{\threeDiameter}{t}(i) \setminus I(u)| \le \badOutEdges |I(u)| \implies |\B{\threeDiameter}{t}(i)| \le \badOutEdges |I(u)| + |I(u)|$$ 

Combining these two relations, and by $|L(u)|\ge \deltaInLp |I(u)|$ (Line~\ref{line:deltaIL} of Algorithm~\ref{algo:lp-cleaning}):
$$|\B{\twoDiameter}{t}(L(u))| \le |\B{\threeDiameter}{t}(i)| \le \frac{(1+\badOutEdges)}{\deltaInLp} |L(u)| = (1+\outsideL) |L(u)|$$
\end{proof}

The following lemma is just a convenient application of the triangle inequality of our LP, that is heavily used in subsequent proofs. Informally, it states that, under certain mild conditions, the LP-distance is small not only if $i,j$ belong in the same LP-cluster (or core-cluster), but even if they happen to be in different clusters that are both intersected by the same third cluster.

\begin{lemma}
\label{lem:threeDiameter}
Let $u,v,w\in \calF$ be three arbitrary nodes. Assume that the LP-cluster of $v$ intersects the LP-clusters of $u$ and $w$ and $t_{max}= \max\{t(u),t(v),t(w)\}$. Then for any $i,j\in \{L(u)\cup L(v)\cup L(w)\}$ their \lpd{t_{max}} is less than $\threeDiameter$, and $\B{\twoDiameter}{t_{max}}(L(u)) \supseteq L(u)\cup L(v)\cup L(w)$.
\end{lemma}
\begin{proof}
If both $i,j$ are either in $L(u)$ or $L(v)$ or $L(w)$ then the claim follows trivially from Lemma~\ref{lem:diameterCluster}. Otherwise we use triangle inequality twice, with species in the intersections of the clusters as intermediates.
More formally, let $k \in L(u) \cap L(v), k' \in L(v) \cap L(w)$.
Lemma~\ref{lem:diameterCluster} implies three things:
\begin{itemize}
    \item [(1)] $x_{i,k}^{(t_{max})} < \diameterInCluster$, for any $i \in L(u) \cup L(v)$
    \item [(2)] $x_{k,k'}^{(t_{max})} < \diameterInCluster$, as both $k,k' \in L(v)$
    \item [(2)] $x_{k',j}^{(t_{max})} < \diameterInCluster$, for any node $j \in L(v)\cup L(w)$
\end{itemize}

Since the LP-distances $x^{(t_{max})}$ respect the triangle inequality it holds that $x_{i,j}^{(t_{max})} < \threeDiameter$. The claim about the ball of $L(u)$ follows by taking the distance from $k$ to $j$.
\end{proof}

We are now ready to prove the hierarchy-friendly property of the output of LP-Cleaning, as we informally claimed when introducing the algorithm. We claim that two LP-clusters of the same level cannot be intersected by the same lower level LP-cluster.

\begin{lemma}
\label{lem:oneIntersection}
Given two nodes $v,w\in \calF$ on the same level, there is no lower level node $u$ such that $L(u)$ intersects both $L(v)$ and $L(w)$.

In particular, there is also no $C(u)$ intersecting both $L(v)$ and $L(w)$.
\end{lemma}
\begin{proof}
The intuition is that $L(v),L(w)$ are close and thus Algorithm~\ref{algo:lp-cleaning} would discard at least one of them.

Without loss of generality, let $|L(v)|\ge |L(w)|$.
$L(v), L(w)$ are disjoint as they are subsets of different parts of the partition $Q^{(t(v))}$, by Algorithm~\ref{algo:lp-cleaning}.

By Lemma~\ref{lem:threeDiameter} $|\B{\twoDiameter}{t(w)}(L(w))| \ge |L(v)|+|L(w)| \ge 2|L(w)|$, which contradicts Lemma~\ref{lem:outsideL}.
\end{proof}

We finally present a simple lower bound on the LP cost. 

%Informally, it states that the cost of our algorithm for removing a species $i$ from an input-cluster at some level $t$ (LP-Cleaning) is proportional to $cost^{(t)}_{i}$.

\begin{lemma}
\label{lem:costLPcleaning}
Let $C_I\in Q^{(t)}$ be an input-cluster at level $t$, and $C_{LP}$ be the respective LP-cluster from Algorithm~\ref{algo:lp-cleaning}. Fix a species $i \in C_I\setminus C_{LP}$. Then the fractional LP cost $cost_i^{(t)} = \Omega(\delta^{(t)}|C_I|)$.
\end{lemma}
\begin{proof}
There are two reasons for $i$ to be in $C_I\setminus C_{LP}$, by Line~\ref{line:outsideA} of Algorithm~\ref{algo:lp-cleaning}. Either half the species in $C_I$ are at distance at least $\radiusInCluster$ from $i$, or more than $\badOutEdges |C_I|$ species not in $C_I$ are at distance at most $\threeDiameter$ from $i$.

In the first case $cost_i^{(t)}\ge \radiusInCluster \cdot (\frac12|C_I|) \delta^{(t)}$, and in the second case $cost_i^{(t)} \ge (1-\threeDiameter)\cdot \badOutEdges|C_I| \delta^{(t)}$.
\end{proof}

\subsection{Derive-Hierarchy results (Algorithm~\ref{algo:Derive-hierarchy})}
In this section we present several structural results related to our algorithm. 

We start with pointing out that our algorithm ends up with the same output, no matter the order in which we process LP-clusters of the same level. This is due to the input sequence $L^{(*)}$ being hierarchy-friendly.

%\begin{restatable}{remark}{order}
\begin{remark}
The output of Algorithm~\ref{algo:Derive-hierarchy} is the same, irrespective of the order in which LP-clusters of the same level are processed in Line~\ref{line:main-iter}.
\end{remark}
%\end{restatable}
\begin{proof}
For each level, fix any ordering in which LP-clusters of the same level are processed, and run the algorithm. For any $t\in [\ell]$, let $\calF_{t-1}$ be the state of the forest just before processing the first node of level $t$. We show that for any level-$t$ LP-cluster $C_{LP}$ with corresponding node $u$ (that is $t(u)=t$ and $L(u)=C_{LP}$), no matter when it was actually processed due to the ordering we fixed, the effect is the same as if it was the first level-$t$ LP-cluster processed. More formally, let $N(u)$ be the set of children of $u$, and $C_{t-1}(u), C^+_{t-1}(u), N_{t-1}(u)$ be the core-cluster, the extended-cluster and the set of children of $u$ in the case where $C_{LP}$ was the first LP-cluster of level-$t$ to be processed. Then $C(u)=C_{t-1}(u), C^+(u)=C^+_{t-1}(u), N(u)=N_{t-1}(u)$.

The main idea is that if a root $v\in \calF_{t-1}$ has a core-cluster intersecting $L(u)$, then it is still a root just before inserting $u$; else $v$ would have another parent $w$ of level $t$, meaning $C(v)\subseteq L(v)$ would also intersect $C(w)\subseteq L(w)$ (Line~\ref{line:conditionExtension}), which contradicts that $L^{(*)}$ is hierarchy-friendly.

For $u$'s children, we first show that $N(u) \subseteq N_{t-1}(u)$. Suppose this was not true, then there would exist a level-$t$ node $v\in N(u)\setminus N_{t-1}(u)$. That would imply that $u$'s and $v$'s core clusters intersect (Line~\ref{line:conditionExtension}). But, core-clusters are always subsets of their corresponding LP-clusters, and LP-clusters of the same level are disjoint.

Before proving $N_{t-1}(u) \subseteq N(u)$, we need to show that $C(u)=C_{t-1}(u)$. We show it by proving that $L(u)\setminus C(u) = L(u) \setminus C_{t-1}(u)$. If a species $i$ is in $L(u)\setminus C(u)$, then it is in the extended-cluster of some node $v$ processed before $u$ such that their core-clusters do not intersect (Line~\ref{line:conditionRemoval}). If $t(v)=t$, then $i$ is either in $C(v)$ (contradiction as it would then not be in $L(u)$), or in the extended-cluster of one of its children $w$, which we proved are of lower-level. Thus $w$ was a root in $\calF_{t-1}$. Again by $L^{(*)}$ being hierarchy-friendly, $C(w)\subseteq L(w)$ does not intersect $L(u)$, meaning it does not intersect $C_{t-1}(u)\subseteq L(u)$ and so $i$ would also be in $L(u)\setminus C_{t-1}(u)$ (Line~\ref{line:conditionRemoval}). If $t(v)<t$, then $v$ itself was a root in $\calF_{t-1}$. The same argument in reverse order is used to prove that if $i$ is in $L(u)\setminus C_{t-1}(u)$ then it is in $L(u)\setminus C(u)$.

We now see that $N_{t-1}(u) \subseteq N(u)$; that is because if $v\in N_{t-1}(u)$, then $v$ is a root in $\calF_{t-1}$ with a core-cluster intersecting $L(u)$, and by $L^{(*)}$ being hierarchy-friendly it is also a root in $\calF(u)\setminus\{u\}$. As $C(v)$ intersects $C_{t-1}(u)=C(u)$, we get $v\in N(u)$ (Line~\ref{line:conditionExtension}).

Finally, a species $i$ in $C^+_{t-1}(u) \setminus C_{t-1}(u)$ is part of the extended cluster of a node in $N_{t-1}(u)$; this child is still a root in $\calF(u)\setminus\{u\}$ by $L^{(*)}$ being hierarchy-friendly, therefore $i \in C^+(u) \setminus C(u)$. The other way around, a species $i$ in $C^+(u) \setminus C(u)$ is part of the extended cluster of a node in $N(u)$; this child is a root in $\calF_{t-1}$ as $u$ has no children of level $t$, therefore $i \in C^+(u) \setminus C_{t-1}(u)$.
\end{proof}

Next, we prove two claims that we have already mentioned informally while introducing Algorithm~\ref{algo:Derive-hierarchy} (Derive-Hierarchy). First we claim that the incrementally built graph is always a forest and also for any node $u$, its extended-cluster contains exactly the species descending from $u$. We notice that the previously stated results do not require these properties.

\begin{lemma}
\label{lem:extended_cluster_descending_species}
For any $u\in \calF$, $\calF(u)$ is a forest of rooted trees with $|S|$ leaves identified with the species of $S$, and for each $u\in \calF$, $C^+(u)$ is the set of $u$'s descending species.
\end{lemma}
\begin{proof}
We prove this inductively based on the order in which the nodes are added in $\mathcal{F}$. The base case for both the claims follows by the initialization of the forest with $|S|$ leaves identified with the species of $S$.

When we insert a node, it becomes the parent of some of the existing roots, therefore the forest structure is preserved.

Next let $u$ be some node in $\mathcal{F}$ and let $v_1,\dots,v_k$ be the children of $u$. Then by construction all these children nodes are added to $\mathcal{F}$ before $u$, and thus by induction argument for each $v_m$ the set of descending species of $v_m$ is exactly the set $C^+(v_m)$. Now we need to prove the same for node $u$. Note that the set of descending species of $u$ is precisely the set $\bigcup_{m\in[k]}C^+(v_m)$. Moreover by construction $C^+(u)=C(u)\cup (\bigcup_{m\in[k]}C^+(v_m))$. Hence to prove the claim we need to show that $C(u)\subseteq \bigcup_{m\in[k]}C^+(v_m)$. For the sake of contradiction let $w\in C(u)$ be a species such that $w\notin \bigcup_{m\in[k]}C^+(v_m)$. As $\calF(u) \setminus \{u\}$ is a forest, there exists a unique node $r$ which is the root node of the tree of $\calF(u) \setminus \{u\}$ that contains $w$. 
Hence again by induction argument $w\in C^+(r)$.
By our assumption, as $r$ is not a child of $u$, $C(u)\cap C(r)=\emptyset$. Thus Algorithm~\ref{algo:Derive-hierarchy} (Line~\ref{line:conditionRemoval}) sets $C(u) \gets C(u)\setminus C^+(r)$ and hence $w\notin C(u)$, which is a contradiction.  
\end{proof}

This simple lemma alone is enough to prove the following corollaries:
\begin{corollary} \label{cor:partitionExtendedClustersOfRoots}
For any $u\in \calF$, the extended-clusters of the root nodes in $\calF(u)$ form a partition of $S$.
\end{corollary}
\begin{proof}
As $\calF(u)$ is a forest, each species is a descendant of exactly one such root and thus belongs in exactly one such extended cluster.
\end{proof}
\begin{corollary}
\label{cor:validOutput}
The output of our algorithm is a sequence of hierarchical partitions of $S$.
\end{corollary}
\begin{proof}
By Corollary~\ref{cor:partitionExtendedClustersOfRoots} the output of the algorithm is a sequence of partitions of $S$.
To see that the output partitions are hierarchical, notice that if two species are in the same rooted tree at some point in the algorithm, then they are never separated as we only add nodes in the forest.
\end{proof}
\begin{corollary}
\label{cor:disjointRemovalExtension}
For any node $u\in \calF$, the species removed from its LP-cluster and the species inserted in its core cluster are disjoint, $\Delta^-(u) \cap \Delta^+(u) = \emptyset$.
\end{corollary}
\begin{proof}
For the sake of contradiction let $i \in \Delta^-(u) \cap \Delta^+(u)$. Since the extended clusters of root nodes in $\calF(u)\setminus \{u\}$ form a partition, let $v$ be the unique such root for which $i\in C^+(v)$.
Now as $i \in \Delta^-(u)$, $C(v)\cap C(u)=\emptyset$ (Line~\ref{line:conditionRemoval} of Algorithm~\ref{algo:Derive-hierarchy}). But again $i \in \Delta^+(u)$ implies $C(v)\cap C(u)\neq \emptyset$ (Line~\ref{line:conditionExtension} of Algorithm~\ref{algo:Derive-hierarchy}) and both these can never be satisfied together.
\end{proof}
\begin{corollary}
\label{cor:intersection-ancestry}
If two nodes $u,v\in \calF$ do not have an ancestry-relationship, then their extended clusters do not intersect.
\end{corollary}
\begin{proof}
If their extended clusters intersected, then they would have a descending species in common, which implies an ancestry-relationship.
\end{proof}

We also need the following result.

\begin{lemma}
\label{lem:extended_union_of_descendants}
For any node $u\in \calF$, its extended-cluster is equal to the union of the core clusters of all descendant nodes $v$ of $u$.
\end{lemma}
\begin{proof}
We prove this inductively. As a base-case, the claim trivially holds for the $|S|$ initial leaves. For an internal node $u$, let $v_1,\dots,v_k$ be the children of $u$ and let $D(u)$ be the descendant nodes of $u$. Then $D(u)=\cup_{m\in[k]}D(v_m)$. Also, by induction, for each $v_m$, $C^+(v_m)=\cup_{w\in D(v_m)}C(w)$. Now as $C^+(u)=C(u)\cup (\cup_{m\in[k]}C^+(v_m))$ we have $C^+(u)= C(u)\cup (\cup_{w\in D(u)}C(w))$. which proves our claim.
\end{proof}

\subsection{Managing removals and extensions}
Using the developed toolkit of structural results, we are ready to show that for any node $u\in \calF$, all three of the LP-cluster $L(u)$, the core-cluster $C(u)$ and the extended-cluster $C^+(u)$ are similar; more than that, we show lower bounds of the LP cost related to $\Delta^-(u) = L(u)\setminus C(u)$ and $\Delta^+(u) = C^+(u) \setminus C(u)$.

In particular, we claim that the following inequality holds for every $u\in \calF$.

\begin{align}
    |\Delta^+(u)| \le \dPlus |C(u)| \label{ineq:dplus}
\end{align}

We prove this claim inductively, based on the order in which nodes are added in $\calF$. As a base case, we initially create a node $u_i$ for each species $i\in S$ with $C(u_i)=C^+(u_i)=\{i\}$, meaning that $\Delta^+(u_i) = \emptyset$.

For any other node $u$, we argue about the size of its extended-cluster $C^+(u)$ in relation with the core-clusters of its descendants, as suggested by Lemma~\ref{lem:extended_union_of_descendants}. We now partition the descendants of $u$ in three parts and argue about each one of them.

\begin{figure}
    \centering
    \includegraphics[width=400pt, height=200pt]{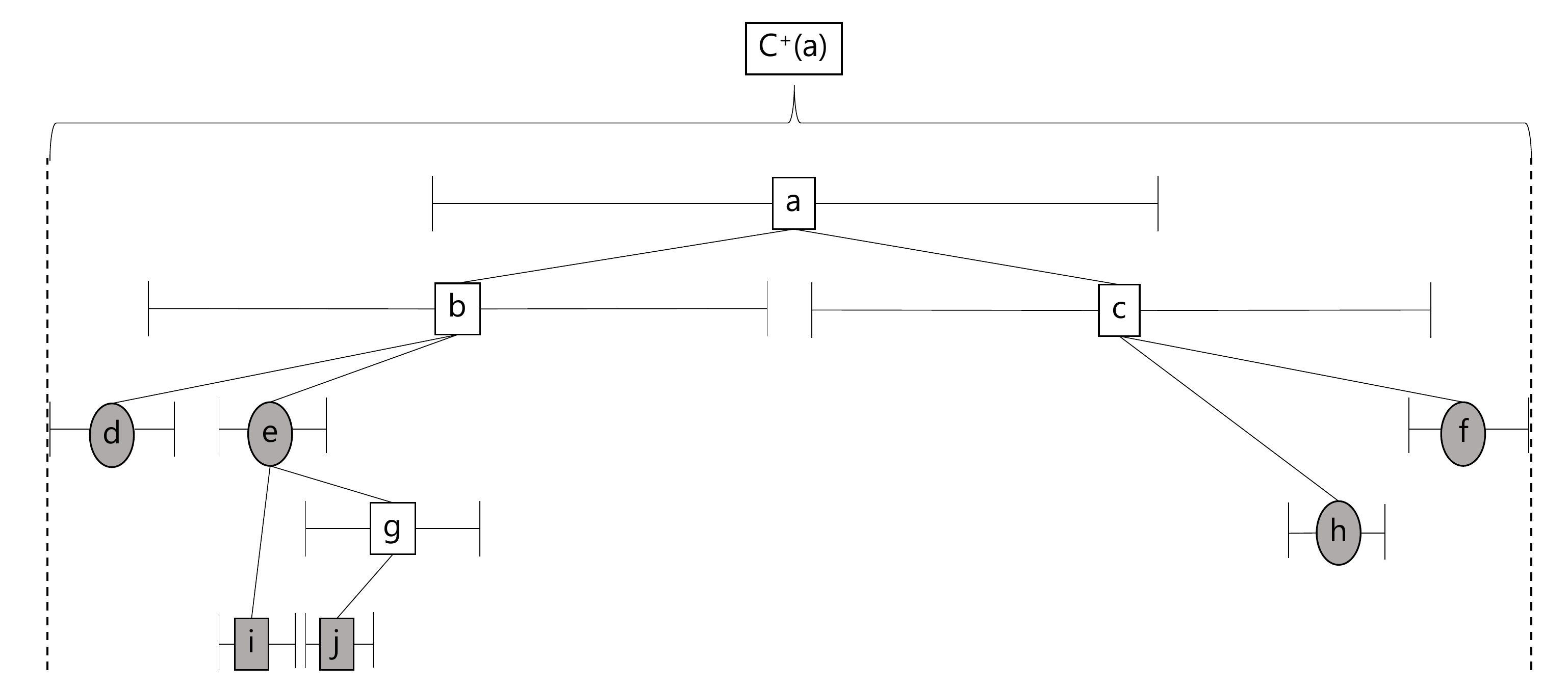}
    \caption{Part of the forest $\calF(\text{a})$. Intervals around nodes denote core-clusters (two core-clusters intersect if a vertical line intersects both); colored nodes $\{d,e,f,h,i,j\}$ have core-clusters not intersecting the core-cluster of $\text{a}$. In particular, the circle-shaped colored nodes are $\text{a}$'s \TNI{} (denoted by $J=\{d,e,f,h\}$). Their proper descendants define $J^+=\{g,i,j\}$. $R=\{b,c\}$ contains all other proper descendants of $\text{a}$.}
    \label{fig:highestNonintersectingDescendants}
\end{figure}

Informally, for a node $u$ we define its \TNI~$J$ as the set of highest level descendant nodes in $\calF$ whose core-clusters do not intersect $C(u)$ (the reader is encouraged to consult Figure~\ref{fig:highestNonintersectingDescendants} before proceeding). More formally, using $v \prec_{\calF} u$ to denote that $v$ is a descendant of $u$ in forest $\calF$, we have:
$$J = \left\{v\in \calF \left| \begin{array}{l}
v \prec_{\calF} u\\ C(u)\cap C(v) = \emptyset\\ C(u)\cap C(w)\not= \emptyset, \forall w \text{ s.t. } v\prec_{\calF}w\prec_{\calF}u \end{array} \right. \right\}$$

Notice that, by definition, if two nodes $v,w$ belong in $u$'s \TNI, then none is an ancestor of the other. Therefore $u$'s \TNI~$J$ naturally partitions the proper descendants of $u$ in three parts: $J$ itself, the set $J^{+}$ of proper descendants of nodes in $J$, and $R$ containing the rest of the proper descendants of $u$ (i.e., the proper descendants of $u$ that are not descendants of any node in $J$). We also define sets of species related to these sets:
\begin{alignat}{2}
    S_{J}  &= \bigcup_{v\in J}     &C(v)  & \label{eq:sj}                                           \\
    S_{J^+}&= \bigcup_{v\in J^{+}} &C(v)  & \quad \setminus (C(u)\cup S_J) \notag            \\
    S_{R}  &= \bigcup_{v\in R}     &C(v)  & \quad \setminus (C(u)\cup S_J\cup S_{J^+})  \notag
\end{alignat}

The apparent asymmetry of not excluding $C(u)$ from $S_J$ follows from the definition of $u$'s \TNI~$J$; the core-clusters of nodes in $J$ are disjoint from $C(u)$, meaning $S_J$ would be the same even if we excluded species in $C(u)$. Note that this is not the case for the core-clusters in $J^{+}$ as proper descendants of nodes in $J$ might still intersect $C(u)$, as in Figure~\ref{fig:highestNonintersectingDescendants}.

Notice that by Lemma~\ref{lem:extended_union_of_descendants} we have $C^+(u) = C(u) \cup (S_J\cup S_{J^+} \cup S_R)$, thus
\begin{align}
    \Delta^+(u) = S_J \cup S_{J^+} \cup S_R \label{eq:deltaPlus}
\end{align}

If $v\in J \cup J^+ \cup R$, and its core-cluster does not intersect the core-cluster of $u$, then by definition of $J$ we have that $v$ is in descendants (not necessarily proper) of $J$. Therefore $v\in J\cup J^+$, meaning that nodes in $R$ have core-clusters that intersect $C(u)$. Furthermore, by definition, each node $v$ in $u$'s \TNI~has a parent whose core-cluster intersects $C(u)$. Therefore, for any species $i\in C(u)$ and any species $j\in S_J \cup S_R$, by Lemma~\ref{lem:threeDiameter} we have that their LP-distance is small, that is 
\begin{align}
    x_{i,j}^{(t(u))} < \threeDiameter, \forall i\in C(u), j\in S_J\cup S_R \label{ineq:small_distance_core_sJ_sR}
\end{align}

Species in $S_J\cup S_R$ are not in $C(u)$, and so by Corollary~\ref{cor:disjointRemovalExtension} they are not in $L(u)$ as they belong in $\Delta^+(u)$. Then Lemma~\ref{lem:outsideL} gives
\begin{align}
    |S_J \cup S_R| < \outsideL |L(u)| \label{ineq:plusBlueGray}
\end{align}

We are left to argue about species in $J^+$, that is in core-clusters of the descendants of $u$'s \TNI. By Lemma~\ref{lem:extended_union_of_descendants}, these species all belong in the extended-clusters of $u$'s \TNI, $\bigcup_{v\in J} C^{+}(v) \supseteq S_J \bigcup S_{J^+}$. By Corollary~\ref{cor:intersection-ancestry} these extended-clusters are disjoint, thus $S_{J^+} \subseteq \bigcup_{v\in J} \Delta^+(v)$. By the inductive hypothesis~(\ref{ineq:dplus}) we get 

\begin{align}
    |S_{J^+}| \le \dPlus |S_J| \label{eq:sizeJplus}
\end{align}

Therefore, by Inequality~(\ref{ineq:plusBlueGray}) we get that 
\begin{align}
    |S_{J^+}| < \outsideL \cdot \dPlus |L(u)|   = \outsideL \cdot \dPlus |C(u)\cup \Delta^-(u)|\label{ineq:plusRed}
\end{align}

By (\ref{ineq:plusBlueGray}) and (\ref{ineq:plusRed}) we bound the size of $\Delta^+(u)$:

\begin{align}
|\Delta^+(u)| < \dPlusOne \cdot \outsideL |L(u)| \label{ineq:extendedByLP}
\end{align}

We are only left with bounding $|L(u)|$. For this we prove that 

\begin{align}
    |\Delta^-(u)| < \dLC |C(u)| \label{ineq:LPbyCore}
\end{align}

which, combined with $(\ref{ineq:extendedByLP})$, proves our initial claim.

\begin{figure}
    \centering
    \includegraphics[width=\textwidth, trim=0.0cm 0.0cm 0.0cm 5.5cm]{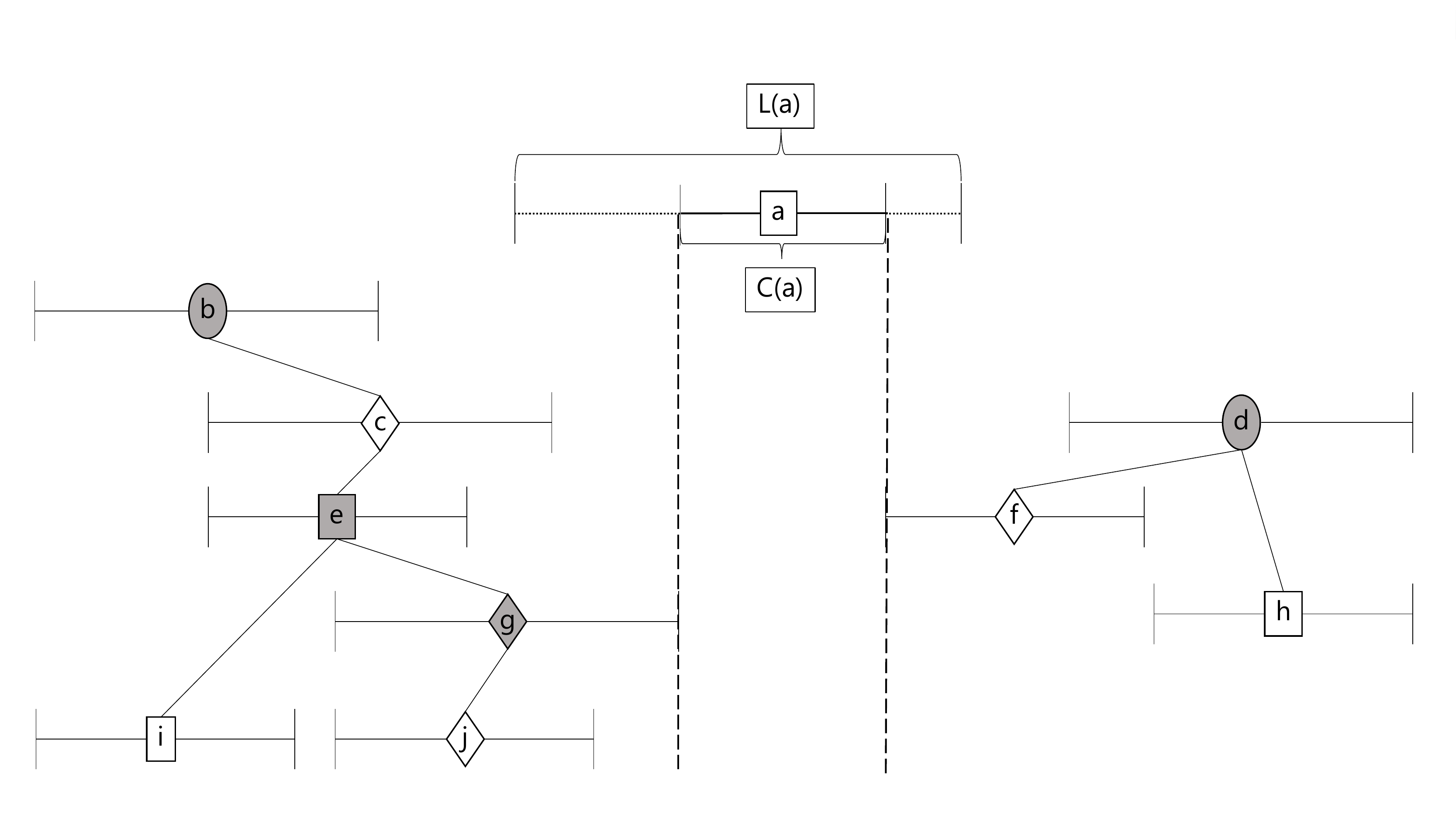}
    \caption{Part of the forest $\calF(\text{a})$. Intervals around nodes denote core-clusters (two core-clusters intersect if a vertical line intersects both); for node $\text{a}$ we also denote its LP-cluster by horizontal dotted lines. All other depicted nodes are not descendants of $\text{a}$. The diamond-shaped nodes $\{c,f,g,j\}$ are contained in $M$, colored nodes $\{b,d,e,g\}$ are contained in $K'$, and circle-shaped nodes $\{b,d\}$ are contained in $K\subseteq K'$.
    }
    \label{fig:highestNonintersectingNonDescendants}
\end{figure}

Before proving $(\ref{ineq:LPbyCore})$ we make some definitions (see Figure~\ref{fig:highestNonintersectingNonDescendants}). 
Roughly speaking, we want to identify an appropriate set $K$ of nodes such that the union of their extended-clusters both contains $\Delta^-(u)$ and its cardinality is reasonably boundable.
In fact, the nodes in $K$ are descendants of roots of $\calF(u)\setminus\{u\}$ that satisfy the condition in Line~\ref{line:conditionRemoval} of Algorithm~\ref{algo:Derive-hierarchy} (i.e., nodes $v$ such that $C(v) \cap C(u) = \emptyset$, and $C^{+}(v) \cap L(u) \not= \emptyset$). 

We now give a formal constructive definition of the set $K$.
Let $M$ be the set containing all the non-descendants of $u$  
%\todo{NP: throwing this long characterisation gives the impression that the reader should understand what it means. E: Shall we also change calligraphic K to something else (say $M$)?} 
at level at most $t(u)$ whose core-clusters intersect $L(u)$. We define $K'$ to be the set of parents of the nodes in $M$. Finally, $K$ is obtained from $K'$ by removing the nodes who have a proper ancestor in $K'$. Notice by Corollary~\ref{cor:intersection-ancestry}, their extended-clusters are disjoint.
We also define sets of species associated with $K$ as follows.
\begin{alignat}{2}
    S_K     &= \bigcup_{v\in K} C(v) \label{ineq:sk}\\
    S_{K^+} &= \bigcup_{v\in K} C^+(v) \notag
\end{alignat}

%work very similarly, by arguing about $v$'s \TNN~$K$. Before digging into details, the reader should first consult Figure~\ref{fig:deleting}. In order to define $v$'s \TNN, it is easier to first define the set $M$ containing the non-descendants of $u$ at level at most $t(v)$ whose core-clusters intersect $L(v)$.

%Then we define $K'$ to contain the parents of nodes in $M$. Finally $K$, the \TNN~of $v$, is just the nodes in $K'$ with no proper ancestor in $K'$; by Corollary~\ref{cor:intersection-ancestry}, their extended-clusters are disjoint.

%We also define sets of species associated with $K$
%\begin{alignat*}{2}
 %   S_K     &= \bigcup_{w\in K} C(w)\\
%    S_{K^+} &= \bigcup_{w\in K} C^+(w)
%\end{alignat*}
Note $\Delta^-(u)\subseteq S_{K^+}$.
Next we claim for each node $v\in K$, $C(v)\cap L(u)=\emptyset$.
Now if we can prove this claim then it implies $S_K\cap L(u)=\emptyset$ and thus we can write $\Delta^-(u)\subseteq S_{K^+}\setminus S_K$.
Next we prove the claim.
Notice that for any node $v \in M$, $v$ is not a descendant of $u$ but $C(v)\cap L(u)\neq \emptyset$; thus there always exists a node $w\in K$ such that $w$ is an ancestor of $v$ and $C(w)\cap L(u)=\emptyset$. 

Now, for the sake of contradiction, assume there exists a node $w\in K$ such that $C(w)\cap L(u)\neq \emptyset$. But then $w\in M$, and following the previous argument there exists a node $w'\in K$ such that $w'$ is an ancestor of $w$ and $C(w')\cap L(u)=\emptyset$. This is a contradiction, as by construction both $w$ and $w'$ cannot be present in $K$.

%that non-descendants of $u$ with a core cluster intersecting $L(v)$ belong in $M$, and thus have a proper ancestor in $v$'s \TNN~$K$. Thus, by Lemma~\ref{lem:extended_union_of_descendants}, $\Delta^-(v) \subseteq S_{K^+} \setminus S_{K}$.

Furthermore, notice that no node $w\in K$ is at level $t(w)= t(u)$, as that would imply a child $w'\in M$ of $w$; but $C(w')$ intersects $C(w)$ (and therefore $L(w)$) as $w'$ is a child of $w$, and $C(w')$ intersects $L(u)$ since $w\in M$. This is a contradiction, by Lemma~\ref{lem:oneIntersection}.

We conclude that $K$ contains nodes at level at most $t(u)-1$, which allows us to apply the inductive hypothesis $|\Delta^+(w)| \le \dPlus |C(w)|$ for nodes $w\in K$. Thus, from $\Delta^-(u) \subseteq S_{K^+} \setminus S_{K}$ we get 

\begin{align}
|\Delta^-(u)| \le \dPlus |S_K| \label{eq:sizeDeltaMinusSk}
\end{align}

Furthermore, all nodes in $K$ have a child whose core-cluster intersects $C(u)$, and so by Lemma~\ref{lem:threeDiameter} the $LP$-distance between a species $i\in L(u)$ and a species $j\in S_K$ is small, $x_{i,j}^{t(u)} < \threeDiameter$. By Lemma~\ref{lem:outsideL} we get $|S_K| < \outsideL |L(u)|$, which gives us $|\Delta^-(u)| \le \outsideL \cdot \dPlus |L(u)|$.

By the definition of $\Delta^-(u)=L(u)\setminus C(u)$ we get $|C(u)| \ge (1-\outsideL \cdot \dPlus)|L(u)|$, by which
\begin{align}
    |\Delta^-(u)| \le \frac{\outsideL \cdot \dPlus}{1-\outsideL \cdot \dPlus} |C(u)| \label{ineq:deltaMinus}
\end{align}
which concludes the proof of claim~(\ref{ineq:LPbyCore}),
and as previously argued, the proof of claim~(\ref{ineq:dplus}).

As a byproduct of this analysis, we can also give some lower bounds on the LP cost. 

\begin{lemma}
\label{lem:lowerBoundDeletion}
Given a node $u\in \calF$, $cost_{I(u)}^{(t(u))} = \Omega(\delta^{(t(u))} |L(u)| |\Delta^-(u)|)$.
\end{lemma}
 \begin{proof}
Fix a $j\in S_K$, as defined in \req{ineq:sk}.
We have shown that for all $i\in L(u)$, the LP-distance with $j$ is small, $x_{i,j}^{(t(u))}<\threeDiameter$.
If $j\in S_K$ is not in the input-cluster $I(u)$, then $cost_{i,j}^{(t(u))} = \delta^{(t(u))} (1-x_{i,j}^{(t(u))}) > \delta^{(t(u))}(1-\threeDiameter)$ for each $\{i,j\}$ pair with $i\in L(u)$.
Else, it was removed from the input-cluster in the LP-Cleaning step, $cost_{j}^{(t(u))}=\Omega(\delta^{(t(u))}|I(u)|)$ by Lemma~\ref{lem:costLPcleaning}.

By the algorithm, $I(u)\supseteq L(u)$, so summing up these costs gives $cost_{I(u)}^{(t(u))} = \Omega(\delta^{(t(u))}|L(u)| |S_K|)$ which is $\Omega(\delta^{(t(u))}|L(u)| |\Delta^-(u)|)$ by $(\ref{eq:sizeDeltaMinusSk})$.
\end{proof}

\begin{lemma}
\label{lem:lowerBoundExtension}
Given a node $u\in \calF$, $cost_{I(u)}^{(t(u))} = \Omega(\delta^{(t(u))}|C(u)| |\Delta^+(u)|)$.
\end{lemma}
\begin{proof}
Let $S_J, S_{J^+}, S_R$ be defined as in \req{eq:sj}. It holds that $\Delta^+(u) = S_J \cup S_{J^+} \cup S_R$ by (\ref{eq:deltaPlus}) and these three sets are pairwise disjoint by definition.
By \req{eq:sizeJplus}, the size of $S_{J^+}$ is small compared to $|S_{J}|$ which implies that $|\Delta^+(u)| = \Omega(|S_J \cup S_R|)$.
Furthermore, by \req{ineq:small_distance_core_sJ_sR}, for any $i\in C(u), j\in S_J\cup S_R$, we have that their LP-distance is small, that is $x_{i,j}^{(t)}<\threeDiameter$. % by Inequality~(\ref{ineq:small_distance_core_sJ_sR}).

We fix such a $j \in S_J \cup S_R$, therefore $j\not \in C(u)$.
If $j\in S_J \cup S_R$ is not in the input-cluster $I(u)$, then $cost_{i,j}^{(t(u))} = \delta^{(t(u))}(1-x_{i,j}^{(t(u))}) > \delta^{(t(u))}(1-\threeDiameter)$ for each $\{i,j\}$ pair with $i\in C(u)$.
Else, $j\in I(u)$, but $j\not\in L(u)$.
That is because, $j$ is not in $C(u)$, and if it was in $L(u)$ then it would contradict Corollary~\ref{cor:disjointRemovalExtension}.
Therefore $j$ was removed from the input-cluster in the LP-Cleaning step (Line~\ref{line:outsideA} of Algorithm~\ref{algo:lp-cleaning}),
and $cost_j^{(t(u))} = \Omega(\delta^{(t(u))}|I(u)|)$ by Lemma~\ref{lem:costLPcleaning}.
Summing these costs proves our claim.
\end{proof}

\subsection{Approximation factor}
In this section we prove that Algorithm~\ref{algo:main-algorithm} is an $\calO(1)$ approximation of the LP cost.

We first make some definitions. Let $t\in [\ell]$. An input-cluster $C_I\in Q^{(t)}$ is \StrongInput{} if there exists a level-$t$ node $u\in \calF$ such that $I(u)=C_I$. Similarly, a part $P$ of the output partition $P^{(t)}$ is \StrongPart{} if there exists a level-$t$ node $u\in \calF$ such that $C^+(u)=P$. In both cases we say that $u$ is the corresponding node. We characterize an input-cluster as \WeakInput{} if it is not \StrongInput{}, and similarly a part of the output partition $P^{(t)}$ as \WeakPart{} if it is not \StrongPart{}.

We start with upper bounding the cost of Algorithm~\ref{algo:main-algorithm}. The upper bound is related to the input-clusters (distinguishing between \StrongInput{} and \WeakInput{}) and the parts of the output partitions (again distinguishing between \StrongPart{} and \WeakPart{}).
Informally, for \WeakInput{} input-clusters and \WeakPart{} parts, the cost of our algorithm is proportional to the sum of squares of their size. For a \StrongInput{} input-cluster with corresponding node $u$, the cost of our algorithm is proportional to its size times the number of species of the input-cluster that did not end up in $u$'s core-cluster. For a \StrongPart{} part of the output partitions, the cost of our algorithm is proportional to its size times the number of its species that did not end up in $u$'s core-cluster.

\begin{lemma}
\label{lem:upperBoundAlgorithm}
Suppose we are given a \NewProblem{} instance $S,Q^{(*)},\delta^{(*)}$ and LP-distances $x^{(*)}$. Then the cost of the output of Algorithm~\ref{algo:main} at level $t$ is at most 
%\todo{NP: I think we should explain better the second term. In particular to what partitions it refers to. E: Is the above paragraph helping in reminding that $P^{(t)}$ is the level-$t$ output-partition?}
\[\delta^{(t)} \left( \sum_{\substack{C_I \in Q^{(t)}\\C_I \text{ is \WeakInput{}} }} {|C_I| \choose 2} + \sum_{\substack{P \in P^{(t)}\\P \text{ is \WeakPart{}}}} {|P| \choose 2} + 
\sum_{\substack{u \in \calF\\ t(u)=t}} \left(|I(u)\setminus C(u)||I(u)| + |\DelPlus u||C^+(u)| \right) \right)\]

\end{lemma}
\begin{proof}
The cost at level $t$ is $\delta^{(t)}$ times the number of pairs $\{i,j\}$ that do not end up in the same part of the output partition $P^{(t)}$ but $\{i,j\} \in \calE(Q^{(t)})$, plus the number of pairs $\{i,j\}$ that end up in the same part of the output partition but $\{i,j\} \not \in \calE(Q^{(t)})$. This is 

\begin{align*}
    \delta^{(t)} &|\calE(Q^{(t)}) \setminus \calE(P^{(t)})| + \delta^{(t)} |\calE(P^{(t)}) \setminus \calE(Q^{(t)})| \\
    = \delta^{(t)} \sum_{C_I\in Q^{(t)}} &|{C_I \choose 2} \setminus \calE(P^{(t)})|
    +
    \delta^{(t)} \sum_{P\in P^{(t)}} |{P \choose 2} \setminus \calE(Q^{(t)})|\\
    = \delta^{(t)} \sum_{\substack{C_I\in Q^{(t)}\\C_I \text{ is \WeakInput{}}}} &|{C_I \choose 2} \setminus \calE(P^{(t)})|
    + \delta^{(t)} \sum_{\substack{C_I\in Q^{(t)}\\C_I \text{ is \StrongInput{}}}} |{C_I \choose 2} \setminus \calE(P^{(t)})| \\
    + \delta^{(t)} \sum_{\substack{P\in P^{(t)}\\P \text{ is \WeakPart{}}}} &|{P \choose 2} \setminus \calE(Q^{(t)})|
    + \delta^{(t)} \sum_{\substack{P\in P^{(t)}\\P \text{ is \StrongPart{}}}} |{P \choose 2} \setminus \calE(Q^{(t)})| \\
\end{align*}

Notice that if $i,j$ are in the same core-cluster $C(u)$ of some node $u$ at level $t(u)=t$, then $\{i,j\} \in \calE(Q^{(t)}) \cap \calE(P^{(t)})$. Also, for each \StrongInput{} input-cluster there exists a corresponding node $u$, and vice-versa (similarly for \StrongPart{} parts of the output-partitions). Therefore:

\begin{align*}
\delta^{(t)} \sum_{\substack{C_I\in Q^{(t)}\\C_I \text{ is \StrongInput{}}}} |{C_I \choose 2} \setminus \calE(P^{(t)})| &\le
\delta^{(t)} \sum_{\substack{u\in \calF\\ t(u)=t}} {|I(u)| \choose 2} - {|C(u)| \choose 2}\\
\delta^{(t)} \sum_{\substack{P\in P^{(t)}\\P \text{ is \StrongPart{}}}} |{P \choose 2} \setminus \calE(Q^{(t)})| &\le
\delta^{(t)} \sum_{\substack{u\in \calF\\ t(u)=t}} {|C^+(u)| \choose 2} - {|C(u)| \choose 2}
\end{align*}

For an input cluster $C_I$ with a corresponding node $u$ at level $t$ such that $I(u)=C_I$, and since always $C(u) \subseteq I(u)$
\[{|I(u)| \choose 2} - {|C(u)| \choose 2} = |I(u)\setminus C(u)||I(u)| - {|I(u)\setminus C(u)| \choose 2} \le |I(u)\setminus C(u)||I(u)|\]
Notice that subtraction is needed since $|I(u)\setminus C(u)||I(u)|$ double-counts the pairs in ${{I(u)\setminus C(u)} \choose 2}$.

Similarly, for a part $P$ of the output partition $P^{(t)}$ with a corresponding node $u$ at level $t$ such that $C^+(u)=P$, and since always $C(u) \subseteq C^+(u)$
\begin{align*}
{|C^+(u)| \choose 2} - {|C(u)| \choose 2} = |C^+(u)\setminus C(u)||C^+(u)| - {|C^+(u)\setminus C(u)| \choose 2} &\le |C^+(u)\setminus C(u)||C^+(u)| \\
&= |\DelPlus u||C^+(u)|
\end{align*}
\end{proof}

For each term of Lemma~\ref{lem:upperBoundAlgorithm}, we show a matching lower bound for the LP cost. First, we give a lower bound of the LP cost related to the \WeakInput{} input-clusters.
\begin{lemma}
\label{lem:lowerBoundCI}
The LP cost at level $t$ $cost^{(t)}$ is
\[\Omega\left(\delta^{(t)} \sum_{\substack{C_I \in Q^{(t)}\\C_I \text{ is \WeakInput{}}}} {|C_I| \choose 2}\right)\]
\end{lemma}
\begin{proof}
Let $C_I \in Q^{(t)}$ be an input-cluster and $C_{LP}$ be the corresponding LP-cluster by Algorithm~\ref{algo:lp-cleaning}. By Lemma~\ref{lem:costLPcleaning}, $cost^{(t)}_{C_{I} \setminus C_{LP}}= \Omega(\delta^{(t)}|C_{I}\setminus C_{LP}||C_I|)$.

If $C_I$ has no corresponding node $u$ with $I(u)=C_I, t(u)=t$, this means that $|C_{LP}|< \deltaInLp |C_I|$, therefore $|C_{I}\setminus C_{LP}|=\Omega(|C_I|)$, which makes the aforementioned cost $\Omega(\delta^{(t)}|C_I|^2) = \Omega(\delta^{(t)}{|C_I| \choose 2})$.

Summing over all these input-clusters may only double-count each pair, which completes the proof.
\end{proof}

Next, we give a lower bound of the LP cost related to the \StrongInput{} input-clusters.
\begin{lemma}
\label{lem:lowerBoundIC}
The LP cost at level $t$ $cost^{(t)}$ is
\[\Omega \left( \delta^{(t)} \sum_{\substack{u \in \calF\\ t(u)=t}} \left(|I(u)\setminus C(u)||I(u)| \right) \right) \]
\end{lemma}
\begin{proof}
Summing the cost of Lemma~\ref{lem:lowerBoundDeletion} over all nodes $u$ at level $t(u)=t$ gives a cost of 
\[\Omega \left( \delta^{(t)} \sum_{\substack{u \in \calF\\ t(u)=t}} \left(|L(u)\setminus C(u)||L(u)| \right) \right) \]
since we may only double-count some pairs. Similarly, summing the cost of Lemma~\ref{lem:costLPcleaning} over all species in such nodes gives a cost of 
\[\Omega \left( \delta^{(t)} \sum_{\substack{u \in \calF\\ t(u)=t}} \left(|I(u)\setminus L(u)||I(u)| \right) \right) \]

We prove our claim by summing these two, and noticing that $|L(u)| \ge \deltaInLp |I(u)|$ by Line~\ref{line:deltaIL} of Algorithm~\ref{algo:lp-cleaning}.
\end{proof}

The following lemma lower bounds the LP cost in relation to the \StrongPart{} parts of the output partition $P^{(t)}$.

\begin{lemma}
\label{lem:lowerBoundCplus}
The LP cost at level $t$ $cost^{(t)}$ is
\[\Omega \left( \delta^{(t)} \sum_{\substack{u \in \calF\\ t(u)=t}} \left(|\DelPlus u||C^+(u)| \right) \right) \]
\end{lemma}
\begin{proof}
Summing the cost of Lemma~\ref{lem:lowerBoundExtension} over all nodes $u$ at level $t(u)=t$ proves our claim, since we may only double-count some pairs.
\end{proof}

The following lemma lower bounds the LP cost in relation to the \WeakPart{} parts of the output partition $P^{(t)}$.
\begin{lemma}
\label{lem:lowerBoundP}
The LP cost at level $t$ $cost^{(t)}$ is
\[ \Omega \left( \delta^{(t)} \sum_{\substack{P \in P^{(t)}\\P \text{ is \WeakPart{}}}} {|P| \choose 2} \right) \]
\end{lemma}
\begin{proof}
Fix any such part $P$. By Algorithm~\ref{algo:Derive-hierarchy}, each part of the output-partition $P^{(t)}$ corresponds to the extended-cluster of some node $u\in \calF$. We use $C_P$ to refer to the core-cluster of this node corresponding to $P$, and notice that $|C_P| = \Omega(|P|)$ by (\ref{ineq:dplus}). Furthermore, by Lemma~\ref{lem:diameterCluster} any two species $i,j\in C_P \subseteq P$ have $x_{i,j}^{(t)} < \diameterInCluster$.

We take two cases based on whether there exists an input-cluster $C_I \in Q^{(t)}$ such that $|C_I \cap C_P| > |C_P/2|$. Since $Q^{(t)}$ is a partition, there may be at most one such $C_I$ for each part $P$. If none exists, then there exist $\Omega({|C_P| \choose 2})$ pairs of species in $C_P$ which belong in different input-clusters, and thus $cost^{(t)}_{C_P} = \Omega(\delta^{(t)}{|C_P| \choose 2}) = \Omega(\delta^{(t)}{|P| \choose 2})$ by \req{ineq:dplus}.

For the remaining parts, we first partition them based on parts that have the same corresponding input-cluster. Let $P_1, \ldots, P_k$ be such a maximal group with the same corresponding input cluster $C_I$, meaning that $|C_I| = \Omega(\sum_{r=1}^{k} |C_{P_r}|) = \Omega(\sum_{r=1}^{k} |P_r|)$. If $C_I$ does not correspond to any node $u$ at level $t$ such that $I(u)=C_I$, then $cost^{(t)}_{C_I} = \Omega(\delta^{(t)}{|C_I| \choose 2})$ by Lemma~\ref{lem:costLPcleaning}, which is $\Omega(\delta^{(t)}\sum_{r=1}^{k} {|P_r| \choose 2})$.

Else there exists such a node $u$ with $I(u)=C_I$, while by the statement of our Lemma there is no $v$ such that $C^+(v)=P_r$ for $r\in [k]$. Therefore all these parts are disjoint from $C^+(u)$ (Corollary~\ref{cor:validOutput}), and thus disjoint from $C(u)$. This implies

$$\bigcup_{r=1}^{k}C_{P_r} \cap I(u) \subseteq I(u)\setminus C(u)$$

Then
$$|I(u)\setminus C(u)| \ge |\bigcup_{r=1}^{k}C_{P_r} \cap I(u)| > \sum_{r=1}^k |C_{P_r}|/2 = \Omega(\sum_{r=1}^k |P_r|)$$

By Lemma~\ref{lem:lowerBoundIC} the LP cost at level $t$ is $\Omega(\delta^{(t)}|I(u)\setminus C(u)||I(u)|) = \Omega(\delta^{(t)}\sum_{r=1}^{k} {|P_r| \choose 2})$.
\end{proof}

We combine all our lower bounds in the following corollary.
\begin{corollary}
\label{cor:lowerBound}
The LP cost at level $t$ is
\[\Omega\left( \delta^{(t)} \left(
\sum_{\substack{C_I \in Q^{(t)}\\C_I \text{ is \WeakInput{}} }} {|C_I| \choose 2} + \sum_{\substack{P \in P^{(t)}\\P \text{ is \WeakPart{}}}} {|P| \choose 2} + 
\sum_{\substack{u \in \calF\\ t(u)=t}} \left(|I(u)\setminus C(u)||I(u)| + |\DelPlus u||C^+(u)| \right) \right)
\right)\]
\end{corollary}
\begin{proof}
Follows by summing the LP cost at level $t$ by Lemmas~\ref{lem:lowerBoundCI}, \ref{lem:lowerBoundIC}, \ref{lem:lowerBoundCplus}, \ref{lem:lowerBoundP}.
\end{proof}

We are now ready to combine all the aforementioned results:
\begin{lemma}
\label{lem:hcaApproximation}
Given a \NewProblem~instance $S,Q^{(*)},\delta^{(*)}$ and LP-distances $x^{(*)}$, the output of
Derive-Hierarchy($S$,LP-Cleaning($S,Q^{(*)},x^{(*)}$))
is a sequence of hierarchical partitions $P^{(*)}$ with cost $\calO(cost^{(*)})$.
\end{lemma}
\begin{proof}
By Corollary~\ref{cor:validOutput}, the output is a sequence of hierarchical partitions of $S$.

For any level $t$, by Lemma~\ref{lem:upperBoundAlgorithm} and Corollary~\ref{cor:lowerBound} the cost of the output of Derive-Hierarchy($S$,LP-Cleaning($S,Q^{(*)},x^{(*)}$)) is within a constant factor from the LP cost. Summing over all levels $t$ proves our lemma.
\end{proof}

Lemma~\ref{lem:hcaApproximation} directly proves:
\begin{equation*}
    \text{\HCA} = \calO(1) \tag*{\req{eq:HCA} from Figure~\ref{fig:roadmap}}
\end{equation*}
as Algorithm~\ref{algo:main-algorithm} simply picks $x^{(*)}$ to be an optimal fractional solution to the LP relaxation.

Combining this with Inequality~\req{eq:reduction_HCC_to_HCA} concludes the proof of Theorem~\ref{thm:hcc}. \qedsymbol

\section{Constant integrality gap} \label{sec:integralityGap}
In this section, we prove that the LP formulation for \GeneralProblem{} (Section~\ref{sec:lp-defs}) has constant integrality gap. This directly extends to the integrality gap of the LP formulation used by Ailon and Charikar for ultrametrics \cite{DBLP:journals/siamcomp/AilonC11}, as the LP formulation for \GeneralProblem~is a generalization of the one for ultrametrics (implicit in \cite{DBLP:journals/siamcomp/AilonC11, 12ac-DBLP:conf/approx/HarbKM05}, discussed in Section~\ref{sec:UltraMetric-HierCorrClust}).

Notice that this is not direct from our algorithm, as for \GeneralProblem{} we do not directly work with the LP from Section~\ref{sec:lp-defs}; we rather reduce our problem to an instance of \NewProblem{}, and then round the LP of this instance.

We start with some definitions.
Suppose we have an instance of Hierarchical Correlation Clustering
$S, \delta^{(*)}=(\delta^{(1)}, \ldots, \delta^{(\ell)}), E^{(*)}=(E^{(1)}, \ldots, E^{(\ell)})$.
We say that $x$ is an LP vector if it consists of LP distances $x_{i,j}$ satisfying the triangle inequality and being in the interval $[0,1]$, for all species $i,j\in S$. For any $E \subseteq {S \choose 2}$, we extend the previously used notion of LP cost as:

\[cost^{(t)}(E,x) = \sum_{\{i,j\} \in E} (x_{i,j}) + \sum_{\{i,j\} \not \in E} (1-x_{i,j})\]

Notice that for any LP vector $x$ and edge-sets $E_1,E_2 \subseteq {S\choose 2}$, we have that 
\begin{align} \label{ineq:propertyCost}
cost^{(t)}(E_1,x) \le cost^{(t)}(E_2,x) + \delta^{(t)}|E_1 \triangle E_2|
\end{align}

That is because only pairs in their symmetric difference may be charged differently by $cost(E_1,x)$ and $cost(E_2,x)$, and the maximum such difference is $\delta^{(t)}$, as the LP-distances are between $0$ and $1$.

The LP formulation of Correlation Clustering, which is a special case of the formulation of \GeneralProblem, has constant integrality gap \cite{DBLP:journals/jcss/CharikarGW05}. 
Therefore, for a Correlation Clustering instance $S,E$, integral solution $Q$ whose cost is within a constant factor from the optimal integral solution $OPT$, and any $t$ and LP vector $x$, it holds that 
\begin{align} \label{ineq:integralityGapCorrClust}
    \delta^{(t)}|E \triangle Q| = \calO(\delta^{(t)}|E\triangle OPT|) = \calO(cost(E,x))
\end{align}

Finally, let $Q^{(*)}=(Q^{(1)},\ldots,Q^{(\ell)})$ be partitions of $S$ such that for each $t\in [\ell]$, $Q^{(t)}$ is a solution to Correlation Clustering with input $S,E^{(t)}$ whose cost is within a constant factor from the optimal. Let $x^{(*)}=(x^{(1)}, \ldots, x^{(\ell)})$ be $\ell$ LP vectors satisfying \req{ineq:hier} that are an optimal fractional solution to \GeneralProblem.

We need to prove that some integral solution $P^{(*)}=(P^{(1)}, \ldots, P^{(l)})$ to \GeneralProblem~has cost within a constant factor of the optimal fractional solution. We pick $P^{(*)}=\text{Derive-Hierarchy}(S,\text{LP-Cleaning}(S,Q^{(*)},x^{(*)}))$, that is the integral solution suggested by Lemma~\ref{lem:hcaApproximation}.
Formally, we prove

\[\sum_{t=1}^{\ell} \delta^{(t)} |P^{(t)} \triangle E^{(t)}| = \calO\left(\sum_{t=1}^{\ell} cost(E^{(t)},x^{(*)}) \right)\]

It holds that 
\[\sum_{t=1}^{\ell} \delta^{(t)} |P^{(t)} \triangle E^{(t)}| \le \sum_{t=1}^{\ell} \delta^{(t)} (|P^{(t)} \triangle \calE(Q^{(t)})| + |\calE(Q^{(t)}) \triangle E^{(t)}|)\]

By Lemma~\ref{lem:hcaApproximation} we have that

\begin{align*}
\sum_{t=1}^{\ell} \delta^{(t)} |P^{(t)} \triangle \calE(Q^{(t)})| &= \calO\left(\sum_{t=1}^{\ell} cost(\calE(Q^{(t)}),x^{(*)}) \right) \\
&\le \calO\left(\sum_{t=1}^{\ell} (cost(E^{(t)},x^{(*)}) + |E\cap \calE(Q^{(t)})| \right)
\end{align*}
with the later following by \req{ineq:propertyCost}. Therefore we bound $\sum_{t=1}^{\ell} \delta^{(t)} |P^{(t)} \triangle E^{(t)}|$:

\begin{align*}
\sum_{t=1}^{\ell} \delta^{(t)} |P^{(t)} \triangle E^{(t)}| &\le \sum_{t=1}^{\ell} \delta^{(t)} (|P^{(t)} \triangle \calE(Q^{(t)})| + |\calE(Q^{(t)}) \triangle E^{(t)}|)\\
&= \calO\left(\sum_{t=1}^{\ell} (cost(E^{(t)},x^{(*)}) + |E^{(t)}\cap \calE(Q^{(t)})| \right)\\
&= \calO\left(\sum_{t=1}^{\ell} cost(E^{(t)},x^{(*)}) \right)
\end{align*}

with the last step following from \req{ineq:integralityGapCorrClust}.

\section{From $L_1$-fitting
ultrametrics to 
hierarchical correlation clustering}\label{sec:UltraMetric-HierCorrClust}
\label{ssec:ultrametric_to_hcc}
For completeness, we here review the reduction
from ultrametrics to hierarchical correlation clustering
implicit in previous work \cite{DBLP:journals/siamcomp/AilonC11, 12ac-DBLP:conf/approx/HarbKM05}.

Given an $L_1$-fitting ultrametrics instance with
input $\mathcal{D}:{S\choose 2}\rightarrow \mathbb R_{>0}$,
we construct an input to the \GeneralProblem{} instance as follows.
Let $D^{(1)}<\ldots< D^{(\ell+1)}$ be the distances that appear in the input distance function $\calD$. For $t=1,\ldots,\ell$, define
\begin{equation}\label{eq:ultrametric-correlation}
\delta^{(t)}=D^{(t+1)}-D^{(t)}\textnormal{ and  }E^{(t)}=\left\{\{i,j\}\in {S\choose 2}\mid \calD(i,j)\le D^{(t)}\right\}
\end{equation}
Now given the solution to this hierarchical correlation clustering
problem, to construct a corresponding ultrametric tree, we first complete
the partition hierarchy with $P^{(0)}$ partitioning $S$ into singletons and $P^{(\ell+1)}$ consisting
of the single set $S$. Moreover, we set $\delta^{(0)}=D^{(1)}$.

To get the ultrametric tree $U$, we create a node for each set in the hierarchical partitioning. Next,
for $t=0,\ldots,\ell$, the parent of a level $t$ node $u$ is the node on level $t+1$ whose set contains the set of $u$, and the length of the parent edge
is $\delta^{(t)}/2$. Then nodes on level $t$ are
of height $\sum_{i=0}^{t-1} \delta^{(t)}/2=D^{(t)}/2$ and
if two species have their lowest common ancestor on level $t$, then their distance is exactly $D^{(t)}$.

The construction is reversible in a manner that given any
ultrametric tree $U$ with leaf set $S$ and all 
distances from $\{D^{(1)}, \ldots, D^{(\ell+1)}\}$, 
we get the
partitions $P^{(1)},\ldots,P^{(\ell)}$ as follows.
First, possibly by introducing nodes with only one child, 
for each species $i$ we make sure it has ancestors of heights $D^{(t)}/2$ for $t=1,\ldots,\ell+1$. Then, for $t=1,\ldots,\ell$, we let partitions $P^{(t)}$ consist of the sets of descendants for each level $t$ node in $U$.

With this relation between $U$ and $P^{(1)},\ldots,P^{(\ell)}$, it follows easily that they have the same cost relative to $\calD$ in the sense that
\[\sum_{t=1}^\ell \delta^{(t)} |E^{(t)}\, \Delta\, E(P^{(t)})|=\sum_{\{i,j\}\in{S\choose 2}} |dist_U(i,j)-\calD(i,j)|.\]
Thus, with \req{eq:ultrametric-correlation}, the hierarchical correlation clustering is equal to $L_1$-fitting ultrametrics with ultrametric distances
from the set of different distances in $\calD$. 

Finally, from Lemma 1(a) in \cite{12ac-DBLP:conf/approx/HarbKM05}, we have that
among all ultrametrics minimizing the $L_1$ distance to $\calD$, there is at least one using only distances from $\calD$. This implies that an $\alpha$-approximation algorithm for hierarchical correlation clustering
implies an $\alpha$-approximation algorithm for
$L_1$-fitting ultrametrics, that is
\begin{equation}\tag*{\req{eq:UltraMetric-HierCorrClust} from Figure~\ref{fig:roadmap}}
    \text{UltraMetric}\leq \text{\HCC}
\end{equation}

Combining this with Theorem~\ref{thm:hcc} concludes the second part of Theorem~\ref{thm:main-result}, namely that the $L_1$-fitting ultrametrics problem can be solved in deterministic polynomial time within a constant approximation factor. \qedsymbol

\section{Tree metric to ultrametric}\label{sec:tree-3ultra}
Agarwala et al. \cite{DBLP:journals/siamcomp/AgarwalaBFPT99} reduced tree metrics to
certain restricted ultrametrics. In fact, their reduction may make certain species have distance $0$ in the final tree, which means that it is actually a reduction from tree pseudometrics\footnote{Pseudometrics are a generalization of metrics that allow distance $0$ between distinct species.} to certain restricted ultrametrics. In this section we show that the restrictions are not needed, and that the reduction can be made in a way that does not introduce zero-distances.

\subsection{Tree pseudometric to (unrestricted) ultrametric}
The claim from \cite{DBLP:journals/siamcomp/AgarwalaBFPT99} is that approximating a certain restricted ultrametric within a factor $\alpha$ can be used to approximate tree pseudometric within a factor $3\alpha$. Here we completely lift these restrictions for $L_1$, and show that they can be lifted for all $L_p$ with $p\in \{2,3,\ldots\}$ with an extra factor of at most $2$.

We will need the well-known characterization of ultrametrics discussed in the introduction, that $U$ is an ultrametric iff it is a metric and $\forall \{i,j,k\} \in {S \choose 3}: U(i,j) \le \max\{U(i,k),U(k,j)\}$.

For simplicity, we prove the theorem only for $p<\infty$, as for $L_{\infty}$ the $3$ approximation \cite{DBLP:journals/siamcomp/AgarwalaBFPT99} cannot be improved by our theorem.

\begin{theorem}
For any integer $1\le p<\infty$, a factor $\alpha\ge 1$ approximation for $L_p$-fitting ultrametrics implies a factor $3 \cdot 2^{(p-1)/p} \cdot \alpha$ approximation for $L_p$-fitting tree pseudometrics.\\
In particular, for $L_1$ it implies a factor $3\alpha$.
\end{theorem}
\begin{proof}[Proof (Extending proof from \cite{DBLP:journals/siamcomp/AgarwalaBFPT99})]
The restriction from Agarwala et al. \cite{DBLP:journals/siamcomp/AgarwalaBFPT99} is 
as follows. For every species $i\in S$, we
have a ``lower bound" $\beta_i$. Moreover,
we have a distinguished species $k\in S$
with an upper bound $\gamma_k$.

We want an ultrametric $U$ such that 
\begin{align*}
    \gamma_k&\geq U(i,j)\geq \max\{\beta_i,\beta_j\} &\forall \{i,j\}\in {S\choose 2}\\
    \gamma_k&=U(k,i) &\forall i \in S\setminus \{k\}.
\end{align*}
We note that the conditions can only be satisfied
if $\gamma_k\geq \beta_i$ for all $i\in S$, so we
assume this is the case. We can even have $\beta_k=\gamma_k$.

The result from \cite{DBLP:journals/siamcomp/AgarwalaBFPT99} states that for any $p$ and $\calD:{S\choose 2}\rightarrow\mathbb R_{>0}$, if we can minimize the restricted ultrametric $L_p$ error within a factor $\alpha$ in polynomial-time,
then there is a polynomial-time algorithm that minimizes the tree pseudometric $L_p$ error within a factor $3\alpha$.

We start with creating a new
distance function $\calD'$.
\[\calD'(i,j)=\min\{ \gamma_k,\max\{\calD(i,j),\beta_i,\beta_j\}\}.\]
Intuitively, we squeeze $\calD'$ to satisfy the restrictions.
For any restricted ultrametric $U$, the error between $U$ and $\calD'$ can never be larger than the error between $U$ and $\calD$, no matter the norm $L_p$. Formally, since $U$ is restricted, we have $\max\{\beta_i,\beta_j\} \le U(i,j) \le \gamma_k$. 
\begin{itemize}
    \item  If $\calD(i,j)>\gamma_k$, then $\calD'(i,j)=\gamma_k\ge U(i,j)$ and $|U(i,j)-\calD'(i,j)|^p < |U(i,j)-\calD(i,j)|^p$.
    \item If $\calD(i,j)<\max\{\beta_i,\beta_j\}$, then $\calD'(i,j)=\max\{\beta_i,\beta_j\}\le U(i,j)$ and $|U(i,j)-\calD'(i,j)|^p < |U(i,j)-\calD(i,j)|^p$.
    \item If $\max\{\beta_i,\beta_j\} \le\calD(i,j)\le \gamma_k$, then $\calD'(i,j)=\calD(i,j)$ and $|U(i,j)-\calD'(i,j)|^p = |U(i,j)-\calD(i,j)|^p$.
\end{itemize}

We now ask for an arbitrary ultrametric fit $U'$ for $\calD'$. With exactly the same reasoning, we can only improve the cost if we replace $U'$ with
\[U(i,j)=\min\{\gamma_k,\max\{U'(i,j),\beta_i,\beta_j)\}\}.\]
Clearly $U$ now satisfies the restrictions (in the end of the proof we show that it is an ultrametric).

Our solution to $L_p$-fitting tree pseudometrics is to first create $\calD'$ from $\calD$, obtain ultrametric $U'$ by an $\alpha$ approximation to $L_p$-fitting ultrametrics, and then obtain the restricted ultrametric $U$ from $U'$. Finally we apply the result from \cite{DBLP:journals/siamcomp/AgarwalaBFPT99} to get the tree pseudometric.

Let $OPT_{\calD,R}$ be the closest restricted ultrametric to $\calD$, and $OPT_{\calD'}$ be the closest ultrametric to $\calD'$. It suffices to show that $\|U-\calD\|_p \le 2^{(p-1)/p}\alpha \|OPT_{\calD,R}-\calD\|_p$ (equivalently $\|U-\calD\|_p^p \le 2^{p-1}\alpha^p \|OPT_{\calD,R}-\calD\|_p^p$, and that $U$ is indeed an ultrametric.

By the above observations, it holds that 
\begin{align*}
   \|\calD'-U\|_p \le \|\calD'-U'\|_p &\le \alpha \|\calD'-OPT_{\calD'}\|_p \le \alpha \|\calD'-OPT_{\calD,R}\|_p \implies \\
   \|\calD'-U\|_p^p &\le \alpha^p \|\calD'-OPT_{\calD,R}\|_p^p
\end{align*}

By definition of $\calD'$, and since $U$ is restricted, for any species $i,j$ it holds
$\min\{\calD(i,j),U(i,j)\}\le \calD'(i,j) \le \max\{\calD(i,j),U(i,j)\}$. The proof follows by a direct case study of the $3$ cases $\calD(i,j)\le \max\{\beta_i,\beta_j\}$, $\max\{\beta_i,\beta_j\} < \calD(i,j) \le \gamma_k$, $\gamma_k < \calD(i,j)$. Therefore

\[|\calD(i,j)-U(i,j)| = |\calD(i,j)-\calD'(i,j)| + |\calD'(i,j)-U(i,j)|\]

For $p\ge 1$ we have
$|x|^p+|y|^p \le (|x|+|y|)^p$,
meaning 
$|\calD(i,j)-\calD'(i,j)|^p + |\calD'(i,j)-U(i,j)|^p \le |\calD(i,j)-U(i,j)|^p$.

Moreover, by the convexity of $|x|^p$ for real $x$, we get 
$((x+y)/2)^p \le (|x|^p+|y|^p)/2$, meaning
$|\calD(i,j)-U(i,j)|^p \le 2^{p-1}(|\calD(i,j)-\calD'(i,j)|^p + |\calD'(i,j)-U(i,j)|^p)$. Therefore

\[|\calD(i,j)-\calD'(i,j)|^p + |\calD'(i,j)-U(i,j)|^p \le  |\calD(i,j)-U(i,j)|^p \le 2^{p-1}(|\calD(i,j)-\calD'(i,j)|^p + |\calD'(i,j)-U(i,j)|^p)\]

The same holds if we replace $U$ with $OPT_{\calD,R}$, as we only used that $U$ is restricted. We now have 

\begin{align*}
\|\calD-U|_p^p &= \sum_{\{i,j\}\in {S\choose 2}} |\calD(i,j)-U(i,j)|^p \\
&\le \sum_{\{i,j\}\in {S\choose 2}} 2^{p-1}(|\calD(i,j)-\calD'(i,j)|^p+|\calD'(i,j)-U(i,j)|^p)\\
&= 2^{p-1}(\sum_{\{i,j\}\in {S\choose 2}} |\calD(i,j)-\calD'(i,j)|^p+\|\calD'-U\|_p^p)\\
&\le 2^{p-1}(\sum_{\{i,j\}\in {S\choose 2}} |\calD(i,j)-\calD'(i,j)|^p+\alpha^p\|\calD'-OPT_{\calD,R}\|_p^p)\\
&\le 2^{p-1}\alpha^p (\sum_{\{i,j\}\in {S\choose 2}} |\calD(i,j)-\calD'(i,j)|^p+\|\calD'-OPT_{\calD,R}\|_p^p)\\
&= 2^{p-1}\alpha^p \sum_{\{i,j\}\in {S\choose 2}}( |\calD(i,j)-\calD'(i,j)|^p+|\calD'(i,j)-OPT_{\calD,R}(i,j)|^p)\\
&\le 2^{p-1}\alpha^p \sum_{\{i,j\}\in {S\choose 2}}( |\calD(i,j)-OPT_{\calD,R}(i,j)|^p)\\
&= 2^{p-1}\alpha^p \|\calD-OPT_{D,R}\|_p^p
\end{align*}

Finally, we need to prove that $U$ inherits that it is an ultrametric. This is clear if we proceed in rounds; each round we construct a new ultrametric, and the last one will coincide with $U$.

More formally, let $U_0=U'$. In the first $|S|$ rounds, we take out a different $i'\in S$ at a time, and let
\[U_r(i,j)=\max\{U_{r-1}(i,j),\beta_{i'}\}.\]
Suppose $r>0$ is the first round where $U_r$ is not an ultrametric. Then there exists a triple $\{i,j,k\}$ such that
$U_r(i,j) > \max\{U_{r-1}(i,k), U_{r-1}(k,j)\}$
As we only increase distances, this may only happen if $U_r(i,j)>U_{r-1}(i,j)$. But this means that $U_r(i,j)=\max\{\beta_i,\beta_j\}$, which is a lower bound on $U_r(i,k)$ and $U_r(k,j)$ by construction.

Finally, $U$ is simply
\[U(i,j)=\min\{\gamma_k,U_{|S|}(i,j)\}.\]
Suppose there exists a triple $\{i,j,k\}$ that now violates the ultrametric property, then it holds that
\[U(i,j) > \max\{U_{|S|}(i,k),U_{|S|}(k,j)\}\]

As we did not increase any distance, this means that both $U(i,k)<U_{|S|}(i,k)$ and $U(k,j)<U_{|S|}(k,j)$; but distances can only reduce to $\gamma_k$ which is an upper bound on $U(i,j)$ by construction.
\end{proof}

\subsection{From tree metric to tree pseudometric}
In this section we prove that in order to find a good tree metric, it suffices to find a good tree pseudometric. This is a minor detail that we add for completeness. Informally, the construction simply replaces $0$ distances with some parameter $\epsilon$, and accordingly adapts the whole metric. By making the parameter $\epsilon$ very small, the cost is not significantly changed.

Technically, our main lemma is the following.

\begin{lemma}
\label{lem:treeMetric_to_treePseudometric}
Given is a set $S$, a distance function $\calD:{S \choose 2} \rightarrow\mathbb{R}_{>0}$, a tree $T$ with non-negative edge weights describing a tree pseudometric on $S$, and a parameter $\alpha\in (0,1]$. In time polynomial in the size of $T$ we can construct a tree $T'$ with positive edge weights describing a tree metric on $S$, such that for any $p\ge 1$, it holds that $\|T'-\calD\|_p \le (1+\alpha)\|T-\calD\|_p$.
\end{lemma}
\begin{proof}
We construct $T'$ from $T$ as follows. First, we contract all edges with weight $0$. This may result in several species from $S$ coinciding in the same node. For each such node $u$ and species $i$ coinciding with some other species in $u$, we create a new leaf-node $u_i$ connected only with $u$ with edge-weight $\epsilon>0$ (to be specified later). We identify $i$ with $u_i$, instead of $u$.

$T'$ describes a tree metric on $S$, as by construction each species $i\in S$ is identified with a distinct node in $T'$, and $T'$ only contains positive edge-weights.

If $T$ matches $\calD$ exactly, that is $\|T-\calD\|_p=0$, then no pair of species $i,j\in S$ have $dist_T(i,j)=0$, as $\calD(i,j)>0$. But then no species coincided in the same node due to the contractions, meaning that no distances changed, which proves our claim. From here on we assume that at least one pair has $dist_T(i,j)\ne \calD(i,j)$.

To specify the parameter $\epsilon$ we first make some definitions. Let $Y$ be the set containing all species $i\in S$ for which we created a new leaf node in $T'$. Moreover, let $d_{min}$ be the smallest positive $|dist_T(i,j)-\calD(i,j)|$ among all $i,j\in S$. Then
\[\epsilon=\alpha d_{min} / (8|S|)\]

For any two species $i,j$, their distance stays the same, increases by $\epsilon$, or increases by $2\epsilon$. Therefore, for $p=\infty$ we directly get $\|T'-\calD\|_p \le \|T-\calD\|_p+2\epsilon$. By definition of $d_{min}$ we have also have $\|T-\calD\|_p\ge d_{min} \implies 2\epsilon \le \alpha \|T-\calD\|_p/(4|S|)<\alpha \|T-\calD\|_p$, which proves our claim. Therefore we can assume that $p<\infty$.

We start with a lower bound related to $\|T-\calD\|_p$. By definition of $Y$, for any $i\in Y$ there exists a $j\in Y$ such that $dist_T(i,j)=0$, meaning that $|dist_T(i,j)-\calD(i,j)|=|\calD(i,j)|\ge d_{min}$. Therefore
\[\|T-\calD\|_p^p \ge \frac{|Y|}2 d_{min}^p\]

We now upper bound $\|T'-\calD\|_p$. If $dist_T(i,j)\ne \calD(i,j)$, then $|dist_T(i,j)-\calD(i,j)|\ge d_{min}$ by definition of $d_{min}$. For the rest of the pairs $i,j$, if their distance increased then either $i\in Y$ or $j\in Y$, by construction; thus there are at most $|Y||S|$ such pairs. Using these observations, we take the following three cases:

\begin{align*}
    \sum_{\substack{
    \{i,j\}\in {S\choose 2}\\
    dist_{T}(i,j)\ne\calD(i,j)}}
    &|dist_{T'}(i,j)-\calD(i,j)|^p
    \le 
    \sum_{\substack{
    \{i,j\}\in {S\choose 2}\\
    dist_{T}(i,j)\ne\calD(i,j)}}
    (|dist_{T}(i,j)-\calD(i,j)|^p + |2\epsilon|^p)
    \\    
    \sum_{\substack{
    \{i,j\}\in {S\choose 2}\\
    dist_{T}(i,j)= \calD(i,j)\\
    dist_{T}(i,j)=dist_{T'}(i,j)}}
    &|dist_{T'}(i,j)-\calD(i,j)|^p
    =
    0
    \\
    \sum_{\substack{
    \{i,j\}\in {S\choose 2}\\
    dist_{T}(i,j)= \calD(i,j)\\
    dist_{T}(i,j)<dist_{T'}(i,j)}}
    &|dist_{T'}(i,j)-\calD(i,j)|^p
    \le
    \sum_{\substack{
    \{i,j\}\in {S\choose 2}\\
    dist_{T}(i,j)= \calD(i,j)\\
    dist_{T}(i,j)<dist_{T'}(i,j)}}
    |2\epsilon|^p
    \le 
    |Y||S||2\epsilon|^p
\end{align*}
Adding these $3$ upper bounds $\|T'-\calD\|_p^p$ by 
\[
\sum_{\substack{
    \{i,j\}\in {S\choose 2}\\
    dist_{T}(i,j)\ne\calD(i,j)}}
    (|dist_{T}(i,j)-\calD(i,j)|^p + |2\epsilon|^p)
    +
    |Y||S||2\epsilon|^p
\]

Using our lower bound and the definition of $\epsilon$ \[|Y||S||2\epsilon|^p \le 2^{p+1}|S||\epsilon|^p\|T-\calD\|_p^p / d_{min}^p \le (\alpha/2) \|T-\calD\|_p^p\]

By the definition of $\epsilon$ and $d_{min}$, for $i,j$ such that $dist_T(i,j)\ne \calD(i,j)$ it holds that $(|dist_{T}(i,j)-\calD(i,j)|^p + |2\epsilon|^p) \le (\alpha/2) |dist_{T}(i,j)-\calD(i,j)|^p$. Therefore, we get

\begin{align*}
    \|T'-\calD\|_p^p \le
\sum_{\substack{
    \{i,j\}\in {S\choose 2}\\
    dist_{T}(i,j)\ne\calD(i,j)}}
    (1+\alpha/2)|dist_{T}(i,j)-\calD(i,j)|^p + (\alpha/2) \|T-\calD\|_p^p \\
    = (1+\alpha/2) \|T-\calD\|_p^p + (\alpha/2) \|T-\calD\|_p^p =
    (1+\alpha) \|T-\calD\|_p^p
\end{align*}
\end{proof}

Therefore, for any $p\ge 1$, we can approximate $L_p$-fitting tree metrics by using an approximation to $L_p$-fitting tree pseudometrics. The error is at most $(1+\alpha)$ times the approximation factor of the tree pseudometric, as any tree metric is also a tree pseudometric.

Setting $\alpha=\frac1{|S|}$ and using the result from \cite{DBLP:journals/siamcomp/AgarwalaBFPT99}, we conclude that

\begin{equation*}
\text{TreeMetric}\leq (3+o(1))\cdot \text{UltraMetric} \tag*{\req{eq:tree-3ultra} from Figure~\ref{fig:roadmap}}
\end{equation*}

This concludes the proof of Theorem~\ref{thm:main-result}. \qedsymbol

As a final note, in the case of $L_0$ (that is, we count the number of disagreements between $\calD$ and $T'$) one cannot hope for a similar result. To see this, let $S$ be a set of species, and let $c_1,c_2\in S$ be two special species. The distance between any pair of species is $2$, except if the pair contains either $c_1$ or $c_2$, in which case the distance is $1$. The optimal tree pseudometric simply sets the distance between $c_1$ and $c_2$ to $0$, and preserves everything else ($1$ disagreement).

Any tree metric requires at least $|S|-3$ disagreements: we say that a non-special species is good if it has tree-distance $1$ to both $c_1$ and $c_2$, and bad otherwise. Bad species have distance different than $1$ to at least one special species, while good species have distance less than $2$ with each other; the disagreements minimize at $|S|-3$, when there is either one or two good species.

\iffalse
This is a kind of annoying issue, but with
the above reduction from tree metric to ultrametric, we could, as output get a tree
with some 0 weight edges and 0 distance between
some species, and strictly speaking, this is not a tree metric, but only a pseudo-metric.
In fact, we even have cases where there is no optimal tree metric: $n$ points at distance 2 from two points at distance 1 from each other... optimal pseudometric puts the two points at distance 0... but we can just make the weight 0 non-zero by a slight increase only affecting the approximation factor marginally....

Now note that if the input distance function $\calD$ is positive and someone comes with a tree $T$ with some zero edge weights, then I can always find a tree that is not worse by a constant factor without zero weights edges. More precisely, let $U$ be a subset of vertices spanned by zero weight edges. Let $w$ be the minimum of distances in $\calD$. We now contract all the zero-weight edges between $U$ and from the contracted vertex, we add weight $w/n$
edges to vertices in $S\cap U$. We could also use weights $w$ staying within a constant factor.
\fi

\section{APX-Hardness} \label{sec:apx}
The problems of $L_1$-fitting tree metrics and $L_1$-fitting ultrametrics are regarded as APX-Hard in the literature \cite{DBLP:journals/siamcomp/AilonC11, 12ac-DBLP:conf/approx/HarbKM05}.
However, we decided to include our own versions of these proofs for a multitude of reasons:
First and foremost, \cite{DBLP:journals/siamcomp/AilonC11} attributes the APX-hardness to \cite{Wareham1993}, which is an unpublished Master thesis that is non-trivial to read. Also \cite{12ac-DBLP:conf/approx/HarbKM05} claims that APX-Hardness of $L_1$-fitting ultrametrics follows directly by the APX-Hardness of Correlation Clustering \cite{DBLP:journals/jcss/CharikarGW05}; but this is only true if all the distances in the ultrametric are in $\{1,2\}$.
Second, we think that our proofs are considerably simpler and more direct.
Finally, our constant factor approximation algorithms for these problems make it important to have formal proofs of their APX-Hardness, since the combination settles that a constant factor approximation is best possible in polynomial time unless P$=$NP.

\subsection{$L_1$-fitting ultrametrics}
The correlation clustering problem has been shown to be APX-Hard in \cite{DBLP:journals/jcss/CharikarGW05}. As noted in \cite{DBLP:journals/siamcomp/AilonC11, 12ac-DBLP:conf/approx/HarbKM05} correlation clustering is the same as the $L_1$-fitting ultrametrics in case both the input and the output are only allowed to have distances in $\{1,2\}$. We refer to this problem as $L_1$-fitting $\{1,2\}$-ultrametrics. Therefore the $L_1$-fitting $\{1,2\}$-ultrametrics is also APX-Hard.

For completeness, we sketch this relation here.
Let $E\subseteq {S\choose 2}$ be an instance of correlation clustering, then $\calD(i,j)$ is an instance to $L_1$-fitting $\{1,2\}$-ultrametrics, where
$\calD(i,j)=1$ if $\{i,j\} \in E$, and $\calD(i,j)=2$ otherwise.
Similarly, given  $\calD$ we can obtain $E$ by setting $\{i,j\}\in S$ iff $\calD(i,j)=1$.
Given any solution to correlation clustering (permutation $P$ of $S$), we get a solution $T$ to $L_1$-fitting $\{1,2\}$-ultrametrics with $T(i,j)=1$ if $i,j$ are in the same part of $P$, and $T(i,j)=2$ otherwise. As $T$ is an ultrametric, we are guaranteed that $T(i,j)\le \max\{T(i,k),T(j,k)\}$, therefore if $T(i,k)=T(j,k)=1$, then $T(i,j)=1$ as only distances in $\{1,2\}$ are allowed. Thus distance-$1$ is a transitive relation and $P$ can be obtained by the equivalence classes of species with distance $1$ in $T$.
The observation from \cite{DBLP:journals/siamcomp/AilonC11} is that $|E \triangle \calE(P)| = \|T-\calD\|_1$, which follows by trivial calculations.

The bird's eye view of our approach for showing APX-Hardness of $L_1$-fitting ultrametrics is the following. For the sake of contradiction, we assume that $L_1$-fitting ultrametrics is not APX-Hard. We then show how to solve the $L_1$-fitting $\{1,2\}$-ultrametrics problem in polynomial time within any constant factor greater than $1$, contradicting the fact that it is APX-Hard. The main idea is that we first solve the general $L_1$-fitting ultrametrics problem. Then we apply a sequence of local transformations that converts the general ultrametric to an ultrametric with distances in $\{1,2\}$ without increasing the error. To achieve this, we first eliminate distances smaller than $1$, then eliminate distances larger than $2$, and then eliminate distances in $(1,2)$.

We first prove the following result concerning the local transformations. We remind the reader that an ultrametric $T$ is defined as a metric with the property that for $i,j,k \in S$ we have $T(i,j) \le \max\{T(i,k),T(j,k)\}$.

\begin{lemma} \label{lem:localTransformUltrametric}
Let $S$ be a set of species, $\calD:{S\choose 2} \rightarrow \{1,2\}$ be a distance function with distances only in $\{1,2\}$, and $T$ be a rooted tree such that each species $i\in S$ corresponds to a leaf in $T$ (more than one species may correspond to the same leaf) and all leaves are at the same depth. Then, in polynomial time, we can create a tree $T_{1,2}$ describing an ultrametric with distances only in $\{1,2\}$ such that $\|T_{1,2}-\calD\|_1 \le \|T-\calD\|_1$.
\end{lemma}
\begin{proof}
We set $T'=T$ and apply the following local transformation to $T'$. If $T(i,j)<1$, we set $T'(i,j)=1$. It holds that $\|T'-\calD\|_1 \le \|T-\calD\|_1$ as $\calD(i,j)\ge 1$ and $T(i,j)<1$ implies $|1-\calD(i,j)|<|T(i,j)-\calD(i,j)|$. Furthermore $T'$ still describes an ultrametric. To see this, notice that $\max\{T'(i,k),T'(j,k)\} \ge \max\{T(i,k),T(j,k)\} \ge T(i,j)$ as we do not decrease distances and $T$ is an ultrametric. Therefore if $T'(i,j)> \max\{T'(i,k),T'(j,k)\}$, this means that $T'(i,j)>T(i,j)$. But this only happens if $T'(i,j)=1$, which is a lower bound on $T'(i,k),T'(j,k)$ by construction. This contradicts that $T'(i,j)> \max\{T'(i,k),T'(j,k)\}$, therefore $T'$ describes an ultrametric. Notice that no two species in $S$ coincide in the same node in $T'$ as the minimum distance between any two distinct species is $1$.

Similarly, we set $T''=T'$ and apply the following local transformation to $T''$. If $T'(i,j)>2$, we set $T''(i,j)=2$. It holds that $\|T''-\calD\|_1 \le \|T'-\calD\|_1$ as $\calD(i,j)\le 2$ and $T'(i,j)>2$ implies $|2-\calD(i,j)|<|T'(i,j)-\calD(i,j)|$.
Furthermore $T''$ still describes an ultrametric. To see this, notice that
$T''(i,j) \le T'(i,j) \le \max\{T'(i,k),T'(j,k)\}$. 
If $T''(i,j) > \max\{T''(i,k),T''(j,k)\}$ then $\max\{T''(i,k),T''(j,k)\} < \max\{T'(i,k),T'(j,k)\}$ which only happens if either of $T'(i,k)$ or $T'(j,k)$ dropped to $2$, meaning that $\max\{T''(i,k),T''(j,k)\}=2$. But $2$ is an upper bound on $T''(i,j)$. This contradicts that $T''(i,j)> \max\{T''(i,k),T''(j,k)\}$, therefore $T''$ describes an ultrametric.

Now, by construction, the ultrametric tree describing $T''$ has leaves at depth $1$ (the maximum distance is $2$) and internal nodes at depth between $0$ and $0.5$ (the minimum distance is $1$). If an internal node $u$ has depth $d_u \in (0,0.5)$, let $x_1$ be the number of pairs $\{i,j\}\subseteq {S\choose 2}$ whose nearest common ancestor is $u$ and $\calD(i,j)=1$, and $x_2$ be the number of pairs $\{i,j\}\subseteq {S\choose 2}$ whose nearest common ancestor is $u$ and $\calD(i,j)=2$. If $x_2\ge x_1$, we remove $u$ and connect the children of $u$ directly with the parent of $u$. We still have an ultrametric as we have an ultrametric tree describing the metric. The $L_1$ error is not larger, as the error of $x_2$ pairs drops by twice the absolute difference in depths between $u$ and its parent (their distance increases but does not exceed $2$), and the error of $x_1\le x_2$ pairs increases by the same amount. Otherwise $x_2<x_1$. In this case we increase the depth of $u$ until it coincides with the depth of some of its children, and merge these children with $u$. Similarly with the previous argument, we still have an ultrametric with smaller $L_1$ error.

Each time we apply the above step, we remove at least one node from our tree. Therefore when we can no longer apply this step, we spent polynomial time and acquired an ultrametric $T_{1,2}$ with distances only in $\{1,2\}$ whose $L_1$ error from $\calD$ is $\|T_{1,2}-\calD\|_1 \le \|T''-\calD\|_1 \le \|T'-\calD\|_1 \le \|T-\calD\|_1$.
\end{proof}

\begin{theorem} \label{thm:ultraAPX}
$L_1$-fitting ultrametrics is APX-Hard. In particular, $L_1$-fitting ultrametrics where the input only contains distances in $\{1,2\}$ is APX-Hard.
\end{theorem}
\begin{proof}
Let $\calD:{S\choose 2} \rightarrow \{1,2\}$ be a distance function, $OPT$ be the optimal ultrametric for the $L_1$-fitting ultrametrics problem, and $OPT_{1,2}$ be the optimal ultrametric for the $L_1$-fitting $\{1,2\}$-ultrametrics. We solve this $L_1$-fitting ultrametrics instance in polynomial time and obtain $T$ such that $\|T-\calD\|_1 \le (1+\epsilon)OPT$ for a sufficiently small constant $\epsilon$, as we assumed that $L_1$-fitting ultrametrics is not APX-Hard. Notice that any solution to the $L_1$-fitting $\{1,2\}$-ultrametrics is also a solution to the $L_1$-fitting ultrametrics, meaning that $\|T-\calD\|_1 \le (1+\epsilon) OPT \le (1+\epsilon) OPT_{1,2}$.

Let $T_{1,2}$ be the ultrametric we get from $T$ by applying Lemma~\ref{lem:localTransformUltrametric}.
Then $T_{1,2}$ is a solution to the $L_1$-fitting $\{1,2\}$-ultrametrics instance, and $\|T_{1,2}-\calD\|_1 \le \|T-\calD\|_1 \le (1+\epsilon)OPT_{1,2}$. This contradicts the fact that $L_1$-fitting $\{1,2\}$-ultrametrics is APX-Hard.
\end{proof}

\subsection{$L_1$-fitting tree metrics}
In this section, we show that $L_1$-fitting tree metrics is APX-Hard. Our reduction is based on the techniques used in \cite{5-DAY1987461} to prove NP-Hardness of the same problem. The bird's eye view of our approach is that we solve $L_1$-fitting ultrametrics by solving $L_1$-fitting tree metrics on a modified instance. In this instance we introduce new species having small distance to each other and large distance to the original species. Through a sequence of local transformations, we show that we can modify the tree describing the obtained tree metric so as to consist of a star connecting the new species, and an ultrametric tree connecting the original species (the center of the star and the root of the ultrametric tree are connected by a large edge). This ultrametric would refute APX-Hardness of $L_1$-fitting ultrametrics, in case $L_1$-fitting tree metrics was not APX-Hard.

\begin{theorem}
$L_1$-fitting tree metrics is APX-Hard.
\end{theorem}
\begin{proof}
Let $\calD:{S\choose 2}\rightarrow \{1,2\}$ be an input to $L_1$-fitting ultrametrics, such that all distances in $\calD$ are in $\{1,2\}$. Moreover, let $n=|S|$ and $OPT_{\calD,U}$ be the ultrametric minimizing $\|OPT_{\calD,U}-\calD\|_1$. By Theorem~\ref{thm:ultraAPX}, this problem is APX-Hard. For the sake of contradiction, assume $L_1$-fitting tree metrics is not APX-Hard. 

Let $\epsilon \in (0,1)$ be a sufficiently small constant, and $M=2(1+\epsilon){n\choose 2}+1$ be a large value. We extend $S$ to $S'\supseteq S$ such that $|S'| = 2n$. For $\{i,j\}\in {S\choose 2}$ we set $\calD'(i,j)=\calD(i,j)$. For $i,j\in {S'\setminus S \choose 2}$ we set $\calD'(i,j)=2$. For all other $i,j$ we set $\calD'(i,j)=M$. As we assumed $L_1$-fitting tree metrics not to be APX-Hard, in polynomial time we can compute $T$, a tree metric such that for any other tree metric $T_0$ it holds that $\|T-\calD'\|_1\le (1+\epsilon)\|T_0-\calD'\|$, for sufficiently small $\epsilon$ such that $0<\epsilon<1$.

We first show that each species $k\in S'\setminus S$ has an incident edge contained in all paths from this species to any species in $S$. To do so, we need to upper bound $\|T-\calD'\|_1$. If we make a star whose leaves are the species in $S$ with distance $1$ from the center, a second star whose leaves are the species in $S'\setminus S$ with distance $1$ from the center, and connect the two centers with an edge of weight $M-2$ then only pairs with both species in $S$ may have the wrong distance, and the error for each such pair is at most one. Therefore $\|T-\calD'\|_1 \le (1+\epsilon){n\choose 2}$. This means that if $k\in S'\setminus S$, then in the tree describing $T$ there exists a path $\Pi_k$ starting from $k$ and having weight larger than $1$, such that the path from $k$ to any species $i\in S$ has $\Pi_k$ as a prefix. To see why this is true, notice that otherwise two species $i,j$ would exist such that the paths from $k$ to $i$ and from $k$ to $j$ only share a prefix $\Pi_{i,j}$ of weight $w_{\Pi_{i,j}}\le 1$. But $T(i,k)>M/2$ as otherwise we would have $\|T-\calD'\|_1 \ge |T(i,k)-\calD'(i,k)| \ge M/2 > (1+\epsilon){n\choose 2}$, and similarly $T(j,k)>M/2$. Then $T(i,j)= T(i,k)+T(j,k)-2\cdot w_{\Pi_{i,j}}>M-2$, meaning again $\|T-\calD'\|_1 > |T(i,j)-\calD'(i,j)| > (1+\epsilon){n\choose 2}$.

Using the aforementioned structural property, we show how to modify our tree so that all species in $S$ are close to each other, all species in $S'\setminus S$ are close to each other, but species in $S$ are far from species in $S'\setminus S$. Let $k\in S'\setminus S$ be the species minimizing $\sum_{i\in S} |T(i,k)-\calD'(i,k)|$. We transform the tree describing $T$ by inserting a node $u$ in the path $\Pi_k$ at distance $1$ from $k$, and creating a star with $u$ as its center and all species in $S'\setminus S$ as leaves at distance $1$. Let $T'$ be the resulting tree metric and notice that $\|T'-\calD'\|_1 \le \|T-\calD'\|_1$ because the errors from species in $S'\setminus S$ to species in $S$ did not increase (by definition of $k$), the errors between species in $S'\setminus S$ are exactly zero, and the errors between species in $S$ stay exactly the same (we did not modify the part of the tree formed by the union of paths between species in $S$).

Then, we modify the tree describing $T'$ to obtain $T''$ so that the distance from any species in $S$ to any species in $S'\setminus S$ is $M$. If for any $i\in S$ we have $T'(i,k)\ne M$, we move $i$ in the tree so as to make its distance with $k$ equal to $M$: if $T'(i,k)<M$, we create a new leaf node connected with $i$ with distance $M-T'(i,k)$, and move $i$ to this new leaf node. Else if $T'(i,k)>M$ there exists an $i'$ (possibly by subdividing an edge) in the path from $k$ to $i$ having distance $M$ from $k$ and we move $i$ to this node. Notice that $\|T''-\calD'\|_1 \le \|T'-\calD'\|_1$ because we move each $i\in S$ by $|M-T'(i,k)|$ so that it has zero error with each $k'\in S'\setminus S$, meaning that the error drops by $|S'\setminus S| |M-T'(i,k)| = n|M-T'(i,k)|$ ($|M-T'(i,k)|$ for each $k'\in S'\setminus S$), and increases by at most $(n-1)|M-T'(i,k)|$ ($|M-T'(i,k)|$ for each $i'\in S\setminus \{i\}$).

If we remove all nodes not in a path from $k$ to any $i\in S$ in the tree describing $T''$, then by construction we have a tree $T_{\calD,U}$ rooted at $k$, having leaves identified with the species in $S$, and all leaves having depth $M$. By the above discussion its error is $\|T_{\calD,U}-\calD\|_1 = \|T''-\calD'\|_1$. As some species may coincide in the same nodes, we get an ultrametric $T'_{\calD,U}$ of $S$ having the aforementioned properties so that no two species coincide in the same node, using Lemma~\ref{lem:localTransformUltrametric}.

Notice that $OPT_{\calD,U}$ has maximum distance between species less than $M$; otherwise its error would be at least $M-2$, which is a contradiction to the fact that an ultrametric where all species have distance $1$ has error at most ${n \choose 2}<M-2$. But then we can take the tree describing this optimal ultrametric, connect its root with a node $u$ so that $u$ has distance $M-1$ to all species in $S$, and identify each species $k'\in S'\setminus S$ with a leaf $u'_{k'}$ connected with $u$ with an edge of weight $1$. If the resulting tree metric is $T_1$, then $\|OPT_{\calD,U}-\calD\|_1 = \|T_1-\calD'\|_1$. We conclude that $\|T'_{\calD,U}-\calD\|_1 = \|T_{\calD,U}-\calD\|_1 = \|T''-\calD'\|_1 \le \|T'-\calD'\|_1 \le \|T-\calD'\|_1 \le (1+\epsilon) \|T_1-\calD'\|_1 = (1+\epsilon) \|OPT_{\calD,U}-\calD\|_1$. This contradicts Theorem~\ref{thm:ultraAPX}.
\end{proof}

% {\color{ForestGreen}
\section{Conclusion} \label{sec:conclusions}

We have given the first constant factor approximation for $L_1$-fitting tree metrics, the first improvement on the problem for the last 16 years. This problem was one of the relatively few remaining problems for which obtaining a constant factor approximation or showing hardness was open. Breaking through the best known $O((\log n)(\log \log n))$-approximation had thus been stated as a fascinating open problem.
% \mtcom{I would like to delete the last sentence in read. We have no clue if this is a step in the right direction for any $L_p$ with $p\neq 1$. All we know is that the very important case of $L_1$ is done.}
% {\color{red}
% Hence, this is a major step towards resolving 
% the complexity of tree and ultrametric fitting problems. }

% \mtcom{The hierarchical correlation clustering problem was already implicit in previous work (so we cannot take much credit for it)... but the special case of hierarchical cluster agreement is very interesting and new, worth mentioning.}

Interestingly, our journey brought us to the study of a natural definition of hierarchical cluster agreement that may be of broader interest, in particular to the data mining community where correlation clustering has been a successful objective function and where hierarchical clustering is often 
desired in practice.

% \mtcom{Would drop the red about a better constant factor since we here in submission have made quite clear we do not care... and also, no-one ever gets close to the tine constant factors from general APX-hardness.., so the problem would just be to find a "small" constant. No need to add it here. We can always add it in the camera when we have resolved our relationship to the constant.}
% {\color{red}
% There are different open problems that the reader may find interesting. Obtaining a constant factor approximation matching the APX-hardness bound would be a great contribution, and this will likely entail proving higher APX-hardness bounds too. }

Finding a polynomial time constant factor approximation (or showing that this is hard, e.g., by reduction to unique games) for $L_2$-fitting tree metrics is a great open problem. Recall from Section~\ref{sec:tree-3ultra} that it suffices to focus on approximating the problem of fitting into an arbitrary ultrametric (no need for restricted versions).
Finally, the $O((\log n) (\log \log n))$-approximation algorithm of Ailon and Charikar for the weighted case (where the cost of an edge is weighted by an input edge weight) could potentially be improved to 
$O(\log n)$ without improving multicut, and it would be interesting to do so. Going even further
would require improving the best known bounds for multicut, a notoriously hard problem.

% }

% Things to write:
% \begin{itemize}
%     \item Sum up of results in a more formal way than in the intro?
%     \item $L_2$ as the main open problem? Using our tools, it suffices to focus on (general unrestricted) ultrametric. Maybe also removing the loglogn from the general problem from Charikar.
% \end{itemize}
% }

\clearpage
\bibliographystyle{plain}
\bibliography{references}

\end{document}